%% file: main.tex
\documentclass[acmsmall]{acmart}

\usepackage{mathpartir} 
\usepackage{stmaryrd} 
\usepackage{listings}
\usepackage{enumitem}
\setlist[description]{leftmargin=0.5\parindent}
\usepackage{multicol}
\usepackage{tablefootnote}
\usepackage[toc,page]{appendix}
\usepackage{xcolor}

\setlist[itemize]{parsep=0pt}
\setlist[enumerate]{parsep=0pt}

\AtBeginDocument{%
  }

\setcopyright{cc}
\setcctype{by}
\acmDOI{10.1145/3798207}
\acmYear{2026}
\acmJournal{PACMPL}
\acmVolume{10}
\acmNumber{OOPSLA1}
\acmArticle{99}
\acmMonth{4}
\received{2025-10-09}
\received[accepted]{2026-02-17}

\input{macros}

\begin{document}

\title{Handling Exceptions and Effects with Automatic Resource Analysis}
\author{Ethan Chu}
\authornote{The first two authors contributed equally to the paper.}
\orcid{0009-0005-6041-0313}
\email{ethanchu@cmu.edu}

\author{Yiyang Guo}
\authornotemark[1]
\orcid{0009-0004-3869-4404}
\email{yiyangg@alumni.cmu.edu}

\author{Jan Hoffmann}
\orcid{0000-0001-8326-0788}
\email{jhoffmann@cmu.edu}

\affiliation{%
  \institution{Carnegie Mellon University}
  \city{Pittsburgh}
  \country{USA}
}

\renewcommand{\shortauthors}{Ethan Chu, Yiyang Guo, Jan Hoffmann}

\begin{abstract}
  There exist many techniques for automatically deriving parametric
  resource (or cost) bounds by analyzing the source code of a program.
  These techniques work effectively for a large class of programs and
  language features.
  However, non-local transfer of control as needed for exception or effect
  handlers has remained a challenge.

  This paper presents the first automatic resource bound analysis that
  supports non-local control transfer between exceptions or effects
  and their handlers.
  The analysis is an extension of type-based automatic amortized
  resource analysis (AARA), which automates the potential method of
  amortized analysis.
  It is presented for a simple functional language with lists and
  linear potential functions.
  However, the ideas are directly applicable to richer settings and
  implemented for Standard ML and polynomial potential functions.

  Apart from the new type system for exceptions and effects,
  a main contribution is a novel syntactic type-soundness theorem that
  establishes the correctness of the derived bounds with respect to a
  stack-based abstract machine.
  An experimental evaluation shows that the new analysis is capable of
  analyzing programs that cannot be analyzed by existing methods and
  that the efficiency overhead of supporting exception and effect
  handlers is low.
\end{abstract}



\maketitle

\section{Introduction}

The goal of \emph{automatic resource bound analysis} is to statically
and automatically infer information on the resource use of a program.
In this context, resources are quantities used during the execution of the
program, such as time, energy, or memory.
Such information is useful for users of a software
library~\cite{IyerAC22}, to schedule jobs in a data center, to ensure
safety and security in resource constraint settings such as embedded
systems~\cite{Wilhelm08,CarbonneauxHRZ13} and smart
contracts~\cite{DasBHP19,AlbertCGRR21}, and to automatically detect
performance bugs~\cite{InferWeb}.
There are many existing methods for inferring resource information, including
type
systems~\cite{Jost03,VasconcelosJFH15,CicekBGGH16,AvanziniD17,WangWC17,HandleyVH19},
deriving and solving recurrence
relations~\cite{Wegbreit75,Grobauer01,Albert07,Kincaid2017popl,CutlerLD20}, and
other static analyses~\cite{GulwaniMC09,AvanziniM13,BrockschmidtEFFG14}.
Most techniques are designed to derive (worst-case) upper
bounds. However, some
analyses focus on (best-case) lower bounds~\cite{AlbertGM13,FrohnNHBG16,NgoDFH16}
and expected cost~\cite{NgoCH17,AvanziniMS20,ChatterjeeGMZ24}.
Automatic resource bound analysis is an undecidable problem, but the
work in the area continues to expand the class of bounds that can be
derived, to improve the efficiency of analyses, and to expand the set of
supported language features.

\paragraph{Non-Local Control}

Some of the most challenging programming language features for
automatic resource analysis involve non-local transfer of control,
such as exception handlers,
continuations~\cite{Felleisen88,Reynolds93}, and effect
handlers~\cite{PlotkinP09,BauerP15}.

\emph{Exceptions} are a popular language feature that can redirect control flow
by dynamically transferring control to an \emph{exception handler}, which often
resides in a previously called function.
It is comparatively straightforward to support errors or exceptions
without handlers~\cite{HoffmannW15} when deriving worst-case bounds.
The reason is that early termination with an error reduces the
resource cost in comparison to the non-error case.
Exception handlers are challenging because resource analyses are
generally designed around regular control-flow patterns such as
structural recursion or simple for loops.
We are only aware of one resource analysis that has limited support
for intraprocedural exceptions handlers~\cite{Albert07}, that is,
handlers that catch exceptions raised in the same function call.
We are not aware of a resource analysis that supports general
exceptions where raised exceptions cannot be statically linked to a
handler in the same function.

\emph{Algebraic effects} and the associated paradigm of programming with effect
handlers have gained popularity in the programming language community in the
past few years~
\cite{tang2024modaleffecttypes,DBLP:journals/pacmpl/RooijK25,Pottier23,VilhenaP21}.
They can be seen as a generalization of exceptions where the user
defines a set of operations (or effects), whose implementation is
context-dependent and provided by an effect handler.
Compared to exception handlers, effect handlers provide additional
difficulties for resource analysis:
they can access a \emph{delimited continuation}~\cite{DBLP:conf/tlca/KiselyovS07,Felleisen88}
to resume the computation at an operations' call site
or to provide a new effect handler for the deferred computation.
Analyzing programs with continuations is hard because it is difficult
to predict the dynamic control-flow.
To the best of our knowledge there are no resource analyses that
support continuations, delimited continuations, or effect handlers.

\paragraph{Resource Bounds for Programs with Exception and Effects}

This paper introduces the first resource bound analysis for computational
effects and their handlers.
In addition, it is the first resource bound analysis that fully
supports exception handlers and exceptions that can transfer
control between different functions.
We extend a type-based technique called automatic amortized resource
analysis (AARA)~\cite{HoffmannJ21}, which was introduced by
Hofmann and Jost~\cite{Jost03}.
In AARA, types constitute resource bounds, and type derivations are
certificates of the soundness of the bound with respect to an
operational cost semantics.
It is parametric in the resource and derives high-water mark bounds
for non-monotone resource metrics.
%
%
The analysis supports recursive types~\cite{GrosenKH23}, higher-order
functions~\cite{Jost10}, arrays, and references~\cite{LichtmanH17}.
AARA can be seen as an automation of the potential methods of
amortized analysis, and the use of potential functions is crucial for
the treatment of exception handlers.

\paragraph{\bf Contributions}

We conservatively extend AARA with typing rules for exceptions, computational
effects, and handlers supporting one-shot continuations so that we can account
for resource cost that depends on the non-local transfer of control.
The cost of executing exception or effect handlers is bounded by a function of
the payload of exceptions or effects.
The analysis is parametric in the type of this payload.
In a nutshell, the idea is that a handler establishes a contract that
consists of a potential function, which specifies the amount of
potential that needs to be provided by an exception or effect.
In this article, we focus on linear worst-case bounds for a simple call-by-value
language with computational effects with handlers, higher-order functions,
lists, sums, and products.
However, extending the techniques to more language features and a
larger set of bounds is orthogonal.
In fact, we implemented the technique for recursive types and
polynomial bounds.
We also show how exceptions, exception handlers, and their respective
typing rules can be derived from effects and effect handlers.
  {\em Besides the design of the resource-aware type system, the main technical
    contributions are the soundness proof, the implementation, and an
    experimental evaluation.}

\paragraph{Soundness Proof}

A technical innovation of this work is a novel proof technique for
establishing the soundness of the derived bounds with respect to a
stack-based abstract machine that specifies the cost and behavior of
programs.
Almost all past soundness proofs are based on big-step cost semantics
and expressions in
let-normal-form~\cite{Jost03,HoffmannAH10,HoffmannW15,KahnH21}.
Here we use a small-step style soundness argument that includes
progress and preservation theorems.
The small-step approach with an abstract machine supports
continuations, which are needed to model computational effects.
Instead of let-normal form, we use a modal separation between values and
expression in the manner of fine-grain call-by-value \cite{levy_2003}. This
explicit separation of pure values and possibly effectul computations
streamlines the typing rules, evaluation rules, and the proof.
Another soundness proof for AARA (without exceptions) based on a
small-step semantics has been given by Grosen et
al.~\cite{GrosenKH23}.
The difference is that the proof in this paper is fully syntactic and
does not rely on a logical relation, uses a stack-based abstract
machine, and a modal separation of expressions and values.
Another approach is to use a big-step cost semantics and a logical
relation to prove soundness~\cite{RajaniGDH20}.

\paragraph{Implementation and Experimental Evaluation}

We have implemented the analysis for Standard ML (SML) programs to evaluate the
effectiveness of the analysis, its overhead compared to a prior implementation
of AARA without support for exception handlers or effects, and the quality of the
derived bounds.
The implementation uses the MLton SML compiler~\cite{Weeks06} as a
front end to convert SML programs into an intermediate form.
In contrast to the presentation in the paper, we support bounds that
are polynomials in sizes of values of recursive types.
We evaluate the implemented analysis with 20 benchmarks, which all use
either exceptions and effects. Of these, 14 cannot be analyzed with existing
techniques, as they either handle exceptions or use effects.
We find that the implementation efficiently finds tight worst-case
bounds for these benchmarks and that
our
technique has low overhead compared to existing AARA implementations
like RaML~\cite{HoffmannAH12a}.

\section{Overview}
\label{sec:overview}

\input{overview}

\section{Type System}
\label{sec:statics}

In this section, we present the syntax of a simple functional language
and type rules for deriving resource bounds. Our language starts with
exceptions and their handlers only, which we will extend to support effects
as well. We also briefly discuss type inference.

\input{syntax}

\input{statics}

\subsection{Type Inference}
\label{sec:inference}

The new rules for exceptions and exceptions handlers are designed so
that the existing type inference algorithm for AARA directly applies
to our type system.
The main idea of the type inference is to reduce the problem of
finding resource annotations to off-the-shelf LP solving.

In this paper, we present the rules in a declarative way and to
perform type inferences we have to first reformulate the rules in a
more algorithmic way.
This includes integrating the structural rules into the syntax
directed ones and fixing finite representation for the sets of
function types.
The latter can be achieved by representing the types with a set of
linear constraints.
The new rules do not present additional challenges for designing
the algorithmic type rules.

Using the algorithmic type rules, type inference for AARA consists of
the following steps.
\begin{enumerate}
  \item Perform Hindley–Milner type inference for the structural types.
  \item Annotate the type derivations with yet unknown
        variables for the resource annotations.
  \item Generate linear constraints based on the algorithmic type rules.
  \item Find a solution for the constraints with an LP solver;
        minimizing the initial potential.
\end{enumerate}
This strategy works for both functions and closed expressions. For
functions, the initial potential is the potential of the argument.
More details about the type inference are in the
literature~\cite{Jost03,HoffmannAH12,HoffmannJ21}.

\section{Soundness via K-Machine}

This section contains the novel soundness proof that establishes the
correctness of the derived bounds via progress and preservation using
the K machine, a stack-based abstract machine.

\subsection{K-Machine Cost Semantics}
\label{sec:dynamics}
\input{dynamics}

\subsection{Type Soundness}
\label{sec:soundness}

\input{soundness}

\section{Evaluation}
\label{sec:evaluation}
\input{evaluation}


\vspace{-1ex}
\section{Related Work}
\label{sec:related}

Automatic resource bound analysis has been approached from many
different angles, but works on non-local control transfer are rare.

\vspace{-1.5ex}
\paragraph{Recurrence Relations}

A traditional approach to automatic resource bound analysis, which was first studied by Wegbreit~\cite{Wegbreit75}, is a two-step
process that consists of extracting then solving recurrence
relations.
It has been studied for functional
programs~\cite{Grobauer01,CutlerLD20,DannerL22}, logic
programs~\cite{DebrayGHL94,Klemen0GMH19}, and imperative languages and
object-oriented languages in the COSTA
project~\cite{Albert07,Puebla08,AlbertAGP11a,AlbertGM13,AlbertFR15}
and by using techniques such as abstract interpretation and symbolic
analysis~\cite{FloresH14,Kincaid2017popl,KincaidBBR2017}.

The only system that covers exception handlers is
COSTA~\cite{Albert07,Albert12}.
The analysis operates on Java bytecode and can derive bounds for
functions that have exceptions handlers which catch exceptions raised in the same function call.
The idea is to turn code that raises an exception into a jump to the
handler.
In contrast, this paper focuses on non-local transfer of control where
raising an exception can transfer control to a different function.

\vspace{-1.5ex}
\paragraph{Term Rewriting}

Most prominently, analyses for term rewriting systems focus on
termination~\cite{BaaderN98,GieslTSF04}. However, resource bound
analysis has also been studied. Two prominent implementations are
AProVE~\cite{NoschinskiEG13,BrockschmidtEFFG14,FrohnNHBG16} and
TCT~\cite{AvanziniM13,HirokawaM14}.
It would be possible to model exceptions with non-determinism in term
rewriting, but it is not clear if this abstraction would lead to
satisfactory resource bounds.
Using types (like in this work) is a natural approach to establish a
modular interface between handling and raising exceptions.

\vspace{-1.5ex}
\paragraph{Other Static Analyses}

Other automatic resource analyses for imperative programs are based on
loop analysis and
instrumentation~\cite{GulwaniMC09,GulwaniJK09,GulwaniZ10,BlancHHK10},
ranking functions~\cite{Zuleger11,ChatterjeeFG17}, and similar
techniques~\cite{SinnZV14}.
In contrast to this paper, these works do not handle non-local
transfer of control as needed for exception handlers.

\vspace{-1.5ex}
\paragraph{Soundness via Abstract Machine}

Whereas our work is the first to prove soundness of AARA on a language with
effects by examining the semantics of a stack-based abstract machine, this
approach is well-established in the wider effect handler community
\cite{10.1145/3519939.3523710,10.1145/3622831}. Most closely related is Voigt
et al's work \cite{10.1145/3763155} demonstrating an effect system that can
ensure that resources (like file handles) are acquired and released safely.

\vspace{-1.5ex}
\paragraph{Automatic Amortized Resource Analysis}

Most closely related to this paper is work on AARA. The technique was introduced by Hofmann and Jost~\cite{Jost03} to infer linear
heap-space bounds for first-order functional programs.
In the past two decades, AARA has been extended to many
programming-language
features~\cite{Jost10,HoffmannW15,LichtmanH17,GrosenKH23}, to
non-linear bounds~\cite{HoffmannAH10,KahnH19,HofmannLOMZ22}, and to
more advanced type systems that forgo full
automation~\cite{RajaniGDH20,WangKRPH20}.
Details can be found in a survey by Jost and
Hoffmann~\cite{HoffmannJ21}.
Many works on AARA consider (deterministic) functional programs with
strict evaluation.
However, there are also works on bounds for
parallel~\cite{HoffmannS15,MullerH24} and
lazy~\cite{SimoesVFJH12,jar2017} evaluation, lower
bounds~\cite{NgoDFH16}, and bounds on expected cost of probabilistic
programs~\cite{WangKH20}.

With the exception of Grosen et al.~\cite{GrosenKH23}, prior work on
AARA establishes the soundness of the derived bounds with respect to a
big-step cost semantics.
Grosen et al.\ present a semantic soundness prove using a small-step
semantics and a logical relation.
This paper presents the first syntactic soundness proof using a
small-step cost semantics.
It is also the first paper that establishes soundness for a
stack-based abstract machine and expressions with a modal separation
of computations and values.

Potential-based reasoning has also been integrated with (concurrent)
separation logic~\cite{Atkey10,ChargueraudP19,MevelJP19}, and Hoare
logic~\cite{CarbonneauxHZ15} to reason about imperative programs.
While the implementation of AARA in Resource Aware
ML~\cite{HoffmannAH12a,HoffmannW15} supports exceptions, no
prior work on AARA considers either exception or effect handlers, or even non-local transfer of control, which is the focus of this paper.

\vspace{-1.5ex}
\paragraph{Other Type-Based Systems}
Other type-based approaches for resource bound analysis are based on
sized types~\cite{VasconcelosJFH15,AvanziniD17}, refinement
types~\cite{Crary00,CicekGA15,CicekBGGH16,HandleyVH19,KnothWHP19,WangKRPH20},
indexed types~\cite{WangWC17}, linear dependent
types~\cite{LagoG11,LagoP13a,LagoP13b}, and structural dependent
types~\cite{Danielsson08,NiuSGH22,Niu023,GrodinNSH24}.
Avanzini et al.~\cite{AvanziniLM15} studied code transformations to
simplify bound inference for higher-order functions.
Sekiyama et al.~\cite{sekiyama2023,kawamata2024} studied analysis of programs
with algebraic effects and handlers using refinement types, as well as
verification of programs with delimited continuations.
These works do not include analyses that automatically derive bounds
for programs with exception or effect handlers.

\vspace{-1ex}
\section{Conclusion and Future Work}
\label{sec:concl}

This paper presents an extension of type-based AARA to exceptions, effects,
and their handlers.
The novel type system supports non-local transfer of control, and its
soundness is established using progress and preservation theorems for
an abstract machine with a control stack.
This is a departure from the traditional soundness proof for AARA,
which uses a big-step cost semantics.
We expect that this alternative formulation of soundness will find
further applications.

While it is possible to model exceptions and effects in a
big-step semantics, control stacks enable a simpler and more local
treatment.
Similarly, it is possible to simulate exceptions and effects with sum
types and higher order functions. However, this would require program
transformations that propagate effects and continuations in deep
subexpressions, which would encumber the analysis.
Supporting general (callcc style) continuations using AARA seems
to require a different approach, since continuations enable a program to use
frames on a control stack multiple times.
Other potential avenues for future work include supporting a richer effect type
system, such as with row polymorphism, or allowing multi-shot continuations.

\section*{Data Availability Statement}
In section \ref{sec:evaluation}, evaluation, we discuss how we implemented AARA
with exceptions and effects as a tool named \tool{}, and evaluated it against 21
benchmark programs. We have made the source code of \tool{} and the 21
benchmarks available as an artifact \cite{aara_effects_artifact}, which also
contains a Dockerfile that automatically builds \tool{} into an executable.
Additionally, the artifact contains an executable of RaML \cite{ramlWeb}, a
prior implementation of AARA that we are comparing \tool{} against. We have
submitted the artifact for evaluation.

\begin{acks}
  This material is based upon work supported by Jane Street Group, LLC
  and the \grantsponsor{NSF}{National Science
    Foundation}{https://doi.org/10.13039/100000001} under grants
  no.~\grantnum{NSF}{2311983} and ~\grantnum{NSF}{2525102}.
  Any opinions, findings, and conclusions or recommendations in this
  material are those of the authors and do not necessarily reflect the
  views of the NSF.
\end{acks}

\bibliographystyle{ACM-Reference-Format}
\bibliography{bib/lit.bib,bib/publications.bib}

\newpage
\input{appendix}

\end{document}

%% file: macros.tex
\usepackage{listings}
\definecolor{background_color}{RGB}{255, 255, 255}
\definecolor{string_color}    {RGB}{180, 156,   0}
\definecolor{identifier_color}{RGB}{  0,   0,   0}
\definecolor{keyword_color}   {RGB}{ 64, 100, 255}
\definecolor{comment_color}   {RGB}{  0, 117, 110}
\definecolor{number_color}    {RGB}{ 84,  84,  84}
\newcommand{\tool}[0]{RaML+}

\newcommand{\code}[1]{\textsf{#1}}

\newcommand{\share}[3]{\ensuremath{#1 \mathop{\curlyveedownarrow} (#2, #3)}}

\newcommand{\abt}[1]{\ensuremath{\mathsf{#1}}}

\newcommand{\trivS}[0]{\ensuremath{\left<\right>}}
\newcommand{\pairS}[2]{\ensuremath{\abt{pair}(#1;#2)}}
\newcommand{\inlS}[1]{\ensuremath{\abt{inl}(#1)}}
\newcommand{\inrS}[1]{\ensuremath{\abt{inr}(#1)}}
\newcommand{\funS}[3]{\ensuremath{\abt{fun}(#1.#2.#3)}}
\newcommand{\linfunS}[2]{\ensuremath{\abt{linfun}(#1.#2)}}
\newcommand{\nilS}[0]{\ensuremath{\abt{nil}}}
\newcommand{\consS}[2]{\ensuremath{\abt{cons}(#1;#2)}}
\newcommand{\dcontS}[1]{\ensuremath{\abt{dcont}(#1)}}

\newcommand{\retS}[1]{\ensuremath{\abt{val}(#1)}}
\newcommand{\tickS}[1]{\ensuremath{\abt{tick}\{#1\}}}
\newcommand{\letS}[3]{\ensuremath{\abt{let}(#1;#2.#3)}}
\newcommand{\appS}[2]{\ensuremath{\abt{app}(#1;#2)}}
\newcommand{\linappS}[2]{\ensuremath{\abt{linapp}(#1;#2)}}
\newcommand{\casepairS}[4]{\ensuremath{\abt{casepair}(#1;#2.#3.#4)}}
\newcommand{\casevoidS}[1]{\ensuremath{\abt{casevoid}(#1)}}
\newcommand{\casesumS}[5]{\ensuremath{\abt{casesum}(#1;#2.#3;#4.#5)}}
\newcommand{\caseS}[5]{\ensuremath{\abt{caselist}(#1;#2;#3.#4.#5)}}
\newcommand{\raiseS}[1]{\ensuremath{\abt{raise}(#1)}}
\newcommand{\tryS}[3]{\ensuremath{\abt{try}(#1;#2.#3)}}
\newcommand{\doS}[2]{\ensuremath{\abt{do}[#1](#2)}}

\newcommand{\newhandlerS}[8]{\ensuremath{\abt{handler}(#1;\{#2: #3.#4.#5\}_{#2 \in #6}; #7. #8)}}

\newcommand{\exty}[2]{\ensuremath{#1 \odot #2}}

\newcommand{\unitT}[0]{\ensuremath{\abt{unit}}}
\newcommand{\prodT}[2]{\ensuremath{{#1} \times {#2}}}
\newcommand{\voidT}[0]{\ensuremath{\abt{void}}}
\newcommand{\sumT}[2]{\ensuremath{{#1} + {#2}}}
\newcommand{\funset}{\mathcal{T}}

\newcommand{\funT}[3]{\ensuremath{#1 \to \exty{#2}{#3}}}
\newcommand{\linfunT}[3]{\ensuremath{#1 \multimap \exty{#2}{#3}}}
\newcommand{\liT}[1]{\ensuremath{\abt{L}(#1)}}

\newcommand{\effectT}[3]{\ensuremath{#1 : #2 \Rightarrow #3}}

\newcommand{\float}[0]{\ensuremath{\mathbb{R}}}
\newcommand{\integer}[0]{\ensuremath{\mathbb{Z}}}

\newcommand{\aty}[2]{\ensuremath{\langle {#1},{#2}\rangle}}

\newcommand{\Q}[0]{\ensuremath{\mathbb{Q}}}

\newcommand{\Rule}[4][]{\ensuremath{\inferrule*[lab={#2},#1]{#3}{#4}}}

\newcommand{\isZero}[1]{\ensuremath{#1 = |#1|}}

\newcommand{\init}{\quad{\textsf{initial}}}
\newcommand{\final}{\quad{\textsf{final}}}

\newcommand{\eval}{\triangleright}
\newcommand{\steps}{\mapsto{}}
\newcommand{\return}{\triangleleft}
\newcommand{\exnret}{\blacktriangleleft}

\newcommand{\takes}{\,\triangleleft\!:}

\newcommand{\proves}{\vdash}

\newcommand{\emptyctx}{\cdot}
\newcommand{\emptystk}{\epsilon}

\newcommand{\funty}[2]{{#1}\rightarrow{#2}}
\newcommand{\abs}[1]{\left|{#1}\right|}

\newcommand{\refappendix}{\ref}

%% file: overview.tex
In this section, we describe the basic idea and state-of-the-art of
type-based automatic amortized resource analysis (AARA), and the challenges that have hindered these systems from analyzing programs with exceptions.
We then discuss the main technical innovations of this work, namely
(in Section~\ref{sec:overview-exn}) a modular extension of AARA that can be combined with existing type systems to derive bounds for programs with exception
handlers and (in Section~\ref{sec:overview-effect}) effect handlers.


\subsection{Type-Based Automatic Amortized Resource Analysis}
\label{sec:overview-aara}

Type-based AARA was introduced by Hofmann and Jost~\cite{Jost03}
to automatically derive \emph{linear bounds} on the heap-memory use of
\emph{first-order} functional programs. Its appeal is that
\begin{itemize}
  \item It derives tight symbolic worst-case bounds for many typical functions,
  \item It is proven sound with respect to a cost semantics,
  \item Type derivations are easily-checkable certificates of the derived bounds,
  \item Type (and therefore bound) inference can be reduced to linear programming (LP), and
  \item The integration of types with the physicist's method of amortized
        analysis makes the system naturally compositional.
\end{itemize}

Maybe surprisingly, these properties of type-based AARA have been
preserved in conservative extensions of the original work to arbitrary
user-defined resource metrics~\cite{Jost09}, non-linear
bounds~\cite{HoffmannAH10,KahnH19,HofmannLOMZ22}, general recursive
types~\cite{GrosenKH23}, lower bounds~\cite{NgoDFH16}, and higher-order
functions~\cite{Jost10}.
A more complete discussion can be found in a recent
survey~\cite{HoffmannJ21}.
To focus on the main ideas, we develop the extension of AARA to exceptions and
effects based on an eager, higher-order functional language with
lists and a type system that only encodes linear resource bounds, which we call linear-bound AARA.
However, we claim that this extension can apply to more complex AARA systems,
e.g. one capable of deriving polynomial resource bounds, with only minor changes.
We omit the discussion of said changes from this paper, but as an example, the
evaluation conducted in section \ref{sec:evaluation} uses polynomial AARA
extended with exceptions and effects.

\paragraph{Cost Model}

Various cost and resource models for functional languages have been
defined for monotone resources such as runtime and
non-monotone ones such as space use.
To abstract from a specific cost model, we follow previous work on
resource analysis \cite{GrodinNSH24,Hoffmann11} and equip the language with
a {\sf tick} effect -- we can think of the expression \code{tick q} as
adding the fixed rational number $q$ to some global counter denoting
the cost incurred.
%
Even though the {\sf tick} effect takes a fixed argument,
we can simulate non-constant ticks. For example, we can define a
function \code{tick\_nat} such that \code{tick\_nat n} will ultimately
perform \code{n} ticks. This is achieved by representing \code{n} as a
\code{unit list} of length \code{n}, to which we apply a recursive
function that traverses the list and performs \code{tick 1} at each
element.

The tick resource
model is a well-established and adequate methodology.
%
%
To define a sound cost model for time (or space), it needs to be
justified with respect to some machine model. This is known as a
bounded implementation and was pioneered by Blelloch and Greiner
\cite{NESL95,BlellochG96}. It is well established that constant tick annotations
correspond to the execution of an interpreter on a random-access
machine (RAM), modulo constant factors. Such cost models do not
necessarily incur cost proportional to the theoretical size increase
resulting from substitution during beta reduction
\cite{DalLago2005}. Instead, we can leverage an abstract machine model that
simulates the lambda calculus more efficiently, such as explicit
substitution/sharing as described by Accattoli \cite{Accattoli2019}.

In most of the examples throughout this paper, we will instrument the code with
ticks to measure the number of \emph{evaluated function calls} and
\emph{exception or effect handler evaluations}. This is an interesting resource
metric because in a functional program, the number of function calls and
handler evaluations performed is proportional to its asymptotic runtime
complexity. Additionally, it can be easily extended to be more realistic --
there has been work demonstrating how to select the amounts to {\sf tick} by
such that the derived bounds correspond to the actual clock cycles needed to
execute the compiled code \cite{Jost10}.

An immediate concern readers may have with our chosen resource metric is that it
assumes constant cost whenever a handler is run, rather than a dynamic cost,
which is often the case when locating the correct effect handler to use. Our
type system sidesteps this concerns in two ways.
First, our handlers require explicit forwarding \cite{DBLP:conf/icfp/KammarLO13,
  Hillerstrom2024}, which is to say that each handler must handle {\it all}
effects that can be performed by the computation being handled. For the
effects that it does not morally handle, it still must explicitly handle,
incur cost via {\sf tick}, then re-perform the effect, thus forwarding it.
This ensures that whenever an effect is performed, control flow is always
transferred to the closest enclosing effect handler, instead of having to
dynamically search for a handler that handles the given effect. Then the {\sf
    tick}'s triggered at each forwarding site effectively simulate the
non-constant cost of identifying a suitable handler.
Next, by enforcing a one-shot discipline for delimited continuations in
handlers, we ensure that non-handler stack frames are never executed more than
once, so the runtime can simply store a pointer to the closest enclosing effect
handler instead of having to copy or shuffle around stack frames
\cite{Sivaramakrishnan2021}.
Together, this justifies incurring only constant cost at each handler
evaluation.

\paragraph{Potential Method}
The potential method of amortized analysis was introduced \cite{tarjan85} to
analyze the cost of a sequence of data structure operations.
The key idea is to define a potential function for the data structure and the
amortized cost of an operation. Then one has to demonstrate that for all
possible executions, the amortized cost and the potential carried by the data
structure is sufficient to cover both the actual cost and the potential of the
resulting data structure.
Then initial potential and the sum of the amortized costs is an upper bound on
the resource consumption of the program.

\paragraph{Resource Annotated Types}
AARA is based on the potential method of amortized analysis.
To automate amortized analysis, it fixes the set of possible potential functions
via types. Specifically, types induce potential functions that map data
structures of the given type to potential, given as a non-negative numbers.
Then, typing rules ensure that terms which introduce or eliminate these
types preserve potential in a manner consistent with the potential functions.
Note that in AARA, the amortized cost is always zero, so the cost of a program
is fully bound by the initial potential.

On top of a simple type system with functions, lists, sums, and products, AARA
equips certain types with potential annotations:
\begin{enumerate}
  \item List types are annotated with a constant potential carried by each
        element of the list. For example, the type $L^q(\unitT)$ encodes a unit
        list that carries $q$ potential per element. In other words, it defines
        the potential function $\Phi(v : L^q(\unitT)) = q|v|$.
  \item Each variant of a sum type is annotated with some constant potential
        carried by an injection into that variant. For example, the type
        $\unitT^p + \unitT^q$ encodes a boolean that carries $p$ potential when
        true and $q$ potential when false. Concretely, it defines the potential
        function $\Phi(v : \unitT^p + \unitT^q) = p \text{ if } v =
          \inlS{\trivS}, q \text{ if } v = \inrS{\trivS}$.
  \item In functions, the argument type is annotated with a constant potential
        that must be provided when applying the function, and the result type is
        annotated with a constant potential that gets returned alongside the
        result. For example, a function of type $\tau_1 \to^p_q \tau_2$ requires
        $p$ potential to be provided along with its argument, and returns $q$
        potential along with its result.
\end{enumerate}
\vspace{-2ex}

\begin{figure}[t]
  \centering
  \begin{minipage}[c]{0.9\linewidth}
    \begin{lstlisting}
fun sqdist (v1 : real list, v2 : real list) : real =
   case (v1,v2) of ([],[]) => 0.0
    | (x::xs,y::ys) => (x - y) * (x - y) + (R.tick 1; sqdist (xs, ys))
    | ([],_::_) => raise Emis1
    | (_::_,[]) => raise Emis2\end{lstlisting}
  \end{minipage}
  \vspace{-2ex}
  \caption{The SML function \code{sqdist} computes the square of the Euclidean
    distance between two $n$-dimensional vectors (lists of reals). The function
    \code{R.tick} is used to specify resource use.}
  \vspace{-3ex}
  \label{fig:dist}
\end{figure}

\paragraph{Example}

Consider the function \code{sqdist} in Figure~\ref{fig:dist}, which computes the
square of the Euclidean distance between two $n$-dimensional vectors represented
as lists of reals ($L(\float)$).
It recursively iterates down both lists until at least one of the lists is
empty,
raising exception $\code{Emis}i$ if the $i$-th list is shorter ($i = 1,2$).
Since the resource metric chosen is the \emph{number of evaluated function calls},
the call \code{sqdist(v1,v2)} costs $\min(|\code{v1}|,|\code{v2}|)$, which we
would like to derive using AARA.
Indeed, using the simple and local typing rules of linear-bound AARA, we can justify numerous possible types for \code{sqdist}, such as
\begin{mathpar}
  {L^1(\float) \times L^0(\float)} \to^0_0 \float

  {L^0(\float) \times L^1(\float)} \to^3_3 \float

  {L^{0.5}(\float) \times L^{0.5}(\float)} \to^5_5 \float
\end{mathpar}
Given the function call $\code{sqdist}(v_1,v_2)$, these types express the cost
bounds $|v_1|$, $|v_2| + 3 - 3$, and $0.5|v_1| + 0.5|v_2| + 5 - 5$,
respectively.
All three bounds are upper bounds on the actual cost of
$\min(|\code{v1}|,|\code{v2}|)$. Which bound is the best depends on the
context in which the call $\code{sqdist}(v_1,v_2)$ appears. To facilitate this,
our type system actually expresses function types as a set of possible types.
During type inference, this set is governed by linear constraints that we infer.
For example, \code{sqdist} would be inferred to have the following type:
\vspace{-1ex}
$$
  \code{sqdist} : \{ {L^{q_1}(\float) \times L^{q_2}(\float)} \to^p_{p'} \float %
  \mid q_1 + q_2 \geq 1 \land p \geq p' \}
  \vspace{-1ex}
$$
Then at each callsite of \code{sqdist}, an LP solver will select the best concrete
choices for the potential variables $q_1,q_2,p,p'$ to minimize the cost of the
larger program in which \code{sqdist} is called. In this sense, our type system
supports {\it resource polymorphism}. However, for simplicity's sake, most
examples in this paper will only discuss monomorphic type signatures when typing
functions.

Compositionality is achieved by assigning potential to result types.
For example, if we wish to type the expression
$\code{sqdist}(\code{append}(v_1,v_2),v')$ where $v' : L^0(\float)$, then we must
construct a typing of \code{append} with result type $L^1(\float)$. Then
assuming that \code{append} does not contain \code{tick} expressions (and thus
has cost $0$), we can give it the following type:
\vspace{-1ex}
$$
  \code{append} : L^1(\float) \times L^1(\float) \to^0_0 {L^1(\float)}
  \vspace{-1ex}
$$
This type reflects that appending two lists that both carry one potential per element produces a list carrying one potential per element.

\begin{figure}[t]
  \centering
  \begin{minipage}[c]{0.9\linewidth}
    \begin{lstlisting}
fun distances_1 (vs : real list list, p : real list) : real list =
  List.map (fn v => R.tick 1; sqdist (v, p)) vs
fun distances_2 (vs : real list list, p : real list) : real list =
  let fun f v = (R.tick 1; sqdist (v, p)) handle Emis2 => (R.tick 1; ~1.0)
  in List.map f vs end\end{lstlisting}
  \end{minipage}
  \vspace{-2ex}
  \caption{The functions \code{distances\_1} and \code{distances\_2} in SML. The
    exception handler in \code{distances\_2} incurs a cost of 1 whenever it is
    evaluated. The function call \code{List.map f l} applies \code{f} to each
    element of \code{l}. }
  \vspace{-3ex}
  \label{fig:distances}
\end{figure}

\subsection{Handling Exceptions}
\label{sec:overview-exn}

While there are some AARA systems that support errors (or raising
exceptions), there does not exist a type-based AARA system that
supports handling exceptions, which is arguably the non-trivial part
of supporting exceptions that this paper develops.
%
%
In this subsection, we will analyze the functions \code{distances\_1} and
\code{distances\_2} in Figure~\ref{fig:distances}, which call the previously
discussed \code{sqdist} (from Figure~\ref{fig:dist}) and handle the
$\code{Emis}i$ exceptions that it raises. For the sake of simplicity, we will
ignore the function calls induced by \code{List.map}\footnote{When analyzing
  \code{distances\_1} and \code{distances\_2} in the evaluation section
  \ref{sec:evaluation}, we do count the function calls induced by \code{List.map}.
  This is a simplification for the overview only.}, so our cost metric in the
following discussion is the number of function calls \emph{to \code{sqdist}}
only, as well as the number of \emph{exception handler evaluations}.

Consider the function \code{distances\_1}. It
takes a list \code{vs} of vectors and a point \code{p} (a vector) and calls the
standard list map function \code{List.map} to compute the distance squared of
\code{p} to each vector in the list using the function \code{sqdist}.
Recall that we are not counting the number of function applications induced by
\code{List.map}, so
the cost of an evaluation of \code{distances\_1($[v_1\ldots,v_n]$,p)} is
$n + \sum_{1 \leq i \leq n} \min(|v_i|,|p|)$.
One valid upper bound for this is encoded by the AARA typing
\vspace{-1ex}
\[
  \code{distances\_1}  : {L^1(L^1(\float)) \times L^0(\float)} \to^0_0 L^0(\float),
  \vspace{-1ex}
\]
which corresponds to the bound $n + \sum_{1 \leq i \leq n} |v_i|$. We can
actually arrive at this typing by partially relying on the previously derived
typing $\code{sqdist} : {L^1(\float) \times L^0(\float)} \to^0_0 \float$.
\code{distances\_1} maps \code{sqdist} across its first argument, and
\code{sqdist} requires 1 potential per element of its first argument, so
unsurprisingly, \code{distances\_1} requires 1 potential per {\it inner} list
element of its first argument. Slightly more interestingly, to pay for the very
first call to \code{sqdist} in \code{distances\_1} (\code{R.tick 1} in the
lambda on line 2), \code{disances\_1} must carry 1 additional potential per {\it
    outer} list element.

Next, consider the function \code{distances\_2}.
It performs the same computation as \code{distances\_1} but handles
the exception \code{Emis2}, which can be raised by \code{sqdist}.
A \code{R.tick} expression in the handler specifies an additional cost
of $1$ if the handler is evaluated.
Would we simply treat the handler as a sequential computation,
we would derive the bound $n + \sum_{1 \leq i \leq n} (1+|v_i|)$, or $2n + \sum_{1 \leq i \leq n} |v_i|$, for the
cost of evaluating \code{distances\_2($[v_1\ldots,v_n]$,p)}.
At first sight, this looks like a tight bound. However, the type system
presented in this paper can prove the more precise bound
$n + \sum_{1 \leq i \leq n} |v_i|$ by deriving the following typing,
which is identical to the typing of \code{distances\_1}.
\footnote{In this example, the difference is only a constant
  factor. However, such imprecisions can easily multiply in a larger
  program and lead to asymptotic differences in the derived bounds.}
\vspace{-1ex}
\[
  \code{distances\_2}  : {L^1(L^1(\float)) \times L^0(\float)} \to^0_0 {L^0(\float)}
  \vspace{-1ex}
\]
Before we discuss how we can derive this tight bound, let us first
intuit why it is correct.
First note that the bound $\sum_{1 \leq i \leq n} |v_i|$ is tight if
each inner list $v_i$ is fully traversed in the function \code{sqdist} (let's
informally call this the worst-case).
Now consider the case in which the exception \code{Emis2}
is raised in the function \code{sqdist}.
In this case, the second argument of \code{sqdist} is shorter than the
first, and consequently, the respective inner list $v_i$ has not
been fully traversed.
So we are not in the previously described ``worst-case'' situation
when the handler gets evaluated: an additional cost of $1$ does not
exceed the cost of fully traversing $v_i$ and not handling an
exception.

How can we express this subtlety in a type system that is simple
enough to support type inference?
\emph{The idea is to associate potential annotations with exception types, thus
  establishing a contract between the point in the code at which an exception is
  raised and the point at which it is handled.}
The amount of potential is chosen to fit the cost of the innermost
exception handler, where the potential may be consumed.
In return, the same amount of potential has to be provided each time an
exception is raised.

To see the idea in action, let's try to justify the typing of
\code{distances\_2} shown above. We can mostly repeat the steps we took the type
\code{distances\_1}, but now we must decide what potentials $p, q$ to annotate
the exception type $\code{Emis1}^p +  \code{Emis2}^q$ with. There is no handler
for the exception \code{Emis1}, so it can carry no potential, i.e. we let $p =
  0$. In contrast, the handler for \code{Emis2} performs a \code{R.tick 1}, which
means \code{Emis2} should carry 1 potential, i.e. we let $q = 1$. Thus this exception
handler is well-typed.

But on the other end of the contract, we must ensure that whoever raises
\code{Emis2} within \code{distances\_2} must provide $1$ potential. To achieve
this, we must now construct a new typing for \code{sqdist}, specifically for its
callsite in \code{distances\_2}. Earlier, we typed \code{distances\_1} by
relying on the typing $\code{sqdist} : {L^1(\float) \times L^0(\float)} \to^0_0
  \float$, so let's start there again. The first argument \code{v1} in the
\code{sqdist} gets the type $L^1(\float)$, i.e. it carries $1$ potential per
element.
Next, observe that \code{raise Emis2} occurs in the case when $v_1$ matches
\code{\_::\_}, which means its length, and thus the potential it carries, is
$\geq 1$ and available for \code{raise Emis2}.
We can thus justify the following typing of \code{sqdist}, now equipped with an annotated exception type:
\vspace{-1ex}
\[
  \code{sqdist} :  L^1(\float) \times L^0(\float) \to^0_{0 \odot 0} \float \odot
  (\code{Emis1}^0 +  \code{Emis2}^1)
  \vspace{-1ex}
\]
Here, we have modified annotated function types to become $\tau_1 \to^p_{q \odot
  q_e} \tau_2 \odot \tau_e$, where $\tau_e$ is the resource annotated exception
type, and $q_e$ is additional potential returned alongside the raised
exception. \footnote{In most of our examples, $\tau_e$ is a sum type
  containing the various exceptions declared in the program, in which case $q_e$
  is redundant thanks to the annotations at each variant of a sum type. Our type
  system, however, has no restriction on $\tau_e$. } Therefore, this typing
expresses that in the evaluation of \code{sqdist(v1,v2)}, $|\code{v1}|$ is
sufficient to cover the \code{R.tick}'s within {\it and} provide 1 additional
potential whenever \code{Emis2} is raised.

%
Our type system is also capable of typing exceptions carrying
non-trivial payloads like lists. The subsequent section on handling computational effects will demonstrate such a payload.


\begin{figure}[t]
  \begin{minipage}[c]{0.95\linewidth}
    \begin{multicols}{2}
      \begin{lstlisting}[]
effect Insert : int list => unit
effect Remove : unit => int list option

fun hstore (m : unit -> unit) : unit = (
   m () handle
    return x => (fn store => x)
  | Insert n  k => (R.tick 1;
    fn store => k () (n :: store))
  | Remove () k => (R.tick 1;
    fn store =>
      case store of
        [] => k NONE []
      | x :: xs => k (SOME x) xs
    ) ) []

fun insert_lists (l : int list list) : unit =
  case l of
    [] => ()
  | x::xs => (do[Insert] x; insert_lists xs)

fun consume () : unit =
  case do[Remove] () of
    NONE => ()
  | SOME x => (traverse x; consume ())

fun store_lists (l : int list list) : unit = hstore (
  fn () => insert_lists l; consume ())
\end{lstlisting}
    \end{multicols}
  \end{minipage}
  \vspace{-2ex}
  \caption{Example implementing a stack store with effect handlers. We assume access to \code{traverse : $L(\integer) \to \unitT$}.
  }
  \vspace{-3ex}
  \label{fig:store}
\end{figure}

\subsection{Handling Computational Effects}
\label{sec:overview-effect}

In this section, we  will demonstrate how our type system can analyze programs
that perform and handle effects. We will consider programs equipped with a set
of named effects, each with an input and output type. For example, the canonical
IO effect \code{Print} might have input type \code{string} and output type
\code{unit}. An effect \code{EName} takes a payload \code{x} of its input type,
and is performed via the expression \code{do[EName] x}, which has the effect's
output type. For example, \code{do[Print] "Hello World"} is an expression that
performs the \code{Print} effect with the string \code{"Hello World"} as its
payload, and has type \code{unit}, as that is the output type.

To handle effects, we have effect handler expressions of the form \code{e {\bf
      handle} {\bf return} a => e1 | EName x k => e2 | \ldots} where \code{e} is the
effectful expression being handled. When \code{e} evaluates to a pure value $v$,
then $v$ is substituted for \code{a} in \code{e1} to get the final value.
However, if some effect \code{EName} is performed, the handler branch \code{e2}
is executed, given the effect's payload \code{x} and a (delimited) continuation
\code{k} that allows it to resume execution of \code{e} from the site where
\code{EName} was performed. We work in the deep handler setting of Kammar et al.
\cite{DBLP:conf/icfp/KammarLO13}, so the continuation \code{k}
reinstalls the effect handler.

Consider the example in Figure \ref{fig:store}. We first define two
effects: \code{Insert} requires an integer list payload and returns
unit, and \code{Remove} requires a unit payload and returns either an
integer list or nothing. Together, they intend to implement a store
data structure.
The function \code{hstore} takes a thunk \code{m}, which may perform the effects \code{Insert} and \code{Remove}, and handles them. In both branches of the handler, we use \code{R.tick} to specify a cost of 1 if the handler is evaluated, just as we did for exceptions. This effect handler uses functions to encode state, threading through the state \code{store}. The branch for \code{Insert} is a function which calls the continuation \code{k} on a unit before updating the state \code{store} by prepending the effect payload \code{n}. The branch for \code{Remove} is a function that cases on the current state \code{store}. If \code{store} is empty, it calls the continuation \code{k} on \code{None} followed by passing the empty list as the state. If \code{store} is nonempty, it calls the continuation \code{k} on \code{Some x} where \code{x} is the head of \code{store}, and updates the state to be the tail \code{xs} of \code{store}.

Now we move on to the computation we wish to handle, \code{store\_lists} on line 26. Given a list of lists of integers \code{l}, it first calls \code{insert\_lists l}, which performs the \code{Insert} effect on each element of \code{l}. 
Next, it calls \code{consume ()}, which repeatedly performs the
\code{Remove} effect until it returns \code{NONE}, signifying that the
store has been emptied. Importantly, for each \code{SOME x} that
\code{Remove} outputs, we call \code{traverse x}, which consumes 1
potential per element of \code{x}, i.e. we assume that it has the
typing
\(\textsf{traverse}: L^1(\integer) \to^0_0 \unitT\). Our
type system is then able to derive that \code{store\_lists
  $[v_1, \ldots, v_n]$} has a cost of
$1 + 2n + \sum_{1\le i \le n} |v_i|$, or the following type.
\vspace{-1ex}
\[
  \textsf{store\_lists}: {L^2(L^1(\integer))} \to^1_0 {\unitT}
\]
We start by arguing why this bound is correct. There are three places that incur cost in \code{store\_lists}: The call \code{insert\_lists $[v_1, \ldots, v_n]$} inserts each $v_i$ into the store exactly once, so that incurs $n$ total cost.
The call \code{consume ()} removes elements from the store
until (and including when) it is empty, which results in cost $n + 1$.
Finally, \code{traverse x} is called for every element $x$ removed
from the store, which are $v_1, \ldots, v_n$, incurring a total cost of
$\sum_{1\le i \le n} |v_i|$.


But how is our type system able to figure this out? The answer lies in having annotated types for effects. We build on the type system for exceptions explained in the Subsection~\ref{sec:overview-exn}. Now, the annotated exception type in a function's typing is replaced with an effect signature containing the annotated input and output types of each effect that its body may perform. Similar to annotated exception types, these resource annotated effect signature are contracts between the client code and handler. Let's look at the typings of \code{insert\_lists} and \code{consume} from our example.
\vspace{-0.5ex}
\[
  \code{insert\_lists}:    {L^2(L^1(\integer))} \to^0_0 \exty{{\unitT}}{\{\code{Insert}: {L^1(\integer)} \Rightarrow^2_0 {\unitT}\}}
  \vspace{-0.5ex}
\]
The type of \code{insert\_lists} promises that when the function
performs the \code{Insert} effect, it will supply a list payload that
carries 1 potential per element, as well as 2 additional
potential, while expecting the handler to return 0
potential. Its typing is relatively straightforward -- its argument
\code{$[v_1, \ldots, v_n]$} is a list of lists, each of which carries
$2 + |v_i|$ potential, which it simply forwards to the
\code{Insert} handler via the payload.
\vspace{-0.5ex}
\[
  \code{consume}:  \unitT  \to^1_0 \exty{ \unitT }{\{\code{Remove}: {\unitT} \Rightarrow^1_0 \left( \code{Some}^1(L^1(\integer)) + \code{None}^0 \right) \}}
  \vspace{-0.5ex}
\]
The type of \code{consume} promises to supply 1 potential when the function
performs the \code{Remove} effect and expects in return either \code{Some x}
carrying $1 + |x|$ potential or \code{None} carrying no potential. It pays for
the performance of \code{Remove} from its argument carrying 1 potential. Then
when casing on the result of the \code{Remove}, in the \code{Some} case on line
24, $|x|$ potential pays for the call \code{traverse x}, while the $1$ remaining
potential pays for the recursive call \code{consume ()}.

On the other end of the contract, the effect handler in \code{hstore} must be able to handle the \code{Insert} and \code{Remove} effects with the annotated types above, namely
\vspace{-0.5ex}
\[
  \{\code{Insert}: {L^1(\integer)} \Rightarrow^2_0 {\unitT}, \code{Remove}: {\unitT} \Rightarrow^1_0 \left( \code{Some}^1(L^1(\integer)) + \code{None}^0 \right)\}
  \vspace{-0.5ex}
\]
Furthermore, all branches of an effect handler must have the same type -- in this handler, all branches have type ${L^1(L^1(\integer))} \to^0_0 {\unitT}$, which can be justified as follows.
(1) The \code{return} branch on line 6 is a function that does not
touch its argument, trivially satisfying the type.
(2) The \code{Insert} branch on line 7 starts with $|n| + 2$ potental
units available, where \code{n} is the payload. $1$ potential
pays for \code{R.tick 1}. In the function, \code{store} starts with
type $L^1(L^1(\integer))$. This type is preserved when we prepend
\code{n} along with the remaining $|n| + 1$ potentials to the
store via \code{n :: store}.
(3) The \code{Remove} branch on line 9 starts with 1 potential, which pays for \code{R.tick 1}. In the function, we case on \code{store}, which has type $L^1(L^1(\integer))$ again. If empty, we send \code{None} to the client without any potential. If nonempty, i.e. it matches \code{x::xs}, we send \code{Some x} along with the $|x| + 1$ potential available from this head element to the client.

Just like \code{{\bf raise} v} for exceptions, we are able to transfer complex potential via payload $v$ from \code{{\bf do}[EName] v} to the closest enclosing effect handler. However, unlike exception handlers, effect handlers can also transfer potential back to effectful computations via their delimited continuations.

%% file: syntax.tex
\subsection{Syntax}
We develop our type system on a language with
(higher-order) functions, lists, pairs, and binary sums.
In our implementation, we use a more complete functional language that
includes regular recursive types.
We restrict ourselves to the simple setting because it already
exhibits all the technical difficulties that arise with exception
handlers (and effects).
\vspace{-2ex}
\begin{center}
  \begin{minipage}{0.4\linewidth}
    \begin{align*}
      v ::=~ & x                                 \\
      \mid~  & \funS f x e                       \\
      \mid~  & \trivS      \mid \pairS{v_1}{v_2} \\
      \mid~  & \inlS{v}    \mid \inrS{v}         \\
      \mid~  & \nilS       \mid \consS{v_1}{v_2}
    \end{align*}
  \end{minipage}
  \begin{minipage}{0.4\linewidth}
    \begin{align*}
      e ::=~ & \retS v \mid \letS{e_1}{x}{e_2} \mid \tickS q \mid \appS {v_1} {v_2} \\
      \mid~  & \casepairS{v}{x_1}{x_2}{e}                                           \\
      \mid~  & \casevoidS{v} \mid \casesumS{v}{x_1}{e_1}{x_2}{e_2}                  \\
      \mid~  & \caseS{v}{e_1}{x_1}{x_2}{e_2}                                        \\
      \mid~  & \raiseS v \mid \tryS{e_1}{x}{e_2}
    \end{align*}
  \end{minipage}
\end{center}
We syntactically distinguish values $v$ and expressions/computations $e$ using a
modal separation based on lax logic~\cite{DaviesP01} and fine-grain
call-by-value \cite{levy_2003}. Values include the trivial unit value,
functions, pairs, injections, and lists.
Since we work in a call-by-value setting, variables stand for values
and are therefore considered values.
In the examples and the implementation, we use SML syntax,
which can be translated into the lax syntax.

Subcomponents of expressions are values whenever possible without
restricting expressivity.
They include function application, as well as pattern matching elimination for
pairs, injections, and lists.
Sequential evaluation of computations only occurs within let expressions.
Values $v$ are lifted to expressions using the syntactic form
$\retS v$.
The syntactic form $\tickS q$ specifies the consumption of $q \in \Q$
resources.
As usual, they can be automatically inserted in higher-level code
based on a resource metric that can model runtime, memory use, or
energy consumption.
If $q < 0$ then $q$ resources become available. This is used, for example, to
mark when memory is freed.
Finally, exceptions can be raised using the computation $\raiseS
  v$.
The computation $\tryS{e_1}x{e_2}$ defines an exception handler $e_2$
that handles exceptions raised by $e_1$.


%% file: statics.tex
\subsection{Resource Annotated Types}
\label{sec:statics-types}

Resource annotated types are defined as follows.
\vspace{-1ex}
\begin{mathpar}
  \tau ::= \funset \mid \unitT \mid \prodT{\tau_1}{\tau_2} \mid \voidT \mid \sumT{A_1}{A_2} \mid \liT A\!\!\!\!\!\!

  A, B, C ::= \aty \tau q\!\!\!\!\!\!

  \funset ::= \{\funT{A}{B}{C} \mid \Theta \}
  \vspace{-1ex}
\end{mathpar}
The idea is that types $\tau$ are annotated with non-negative numbers which
define a potential function $\Phi(v : \tau)$ that maps values $v$ of type
$\tau$ to non-negative numbers. The subsequent subsection will formally define the potential functions.

Types \(\tau\) include arrow types, lists, products, and sums. \(A, B, C\) are
types with constant potential \(q \in \Q_{\geq 0}\) attached to them, allowing
us to specify for example \(\liT{\aty \tau q}\), a list of \(\tau\)'s where each
element carries \(q\) potential. Similarly, we can also specify sum types
\(\sumT{\aty {\tau_1}{q_2}}{\aty {\tau_2}{q_2}}\) specifying that a left
injection will carry \(q_1\) potential in addition to having a payload of type
\(\tau_1\), and likewise for the right injection. Finally, function types are
sets \(\funset\) of possible types each of the form \(\funT{A}{B}{C}\) governed
by some predicate \(\Theta\), typically the conjunction of various linear
constraints regarding potential annotations \(q \in \Q_{\ge 0}\) that appear
within \(A, B, C\). Note that $B$ is the function's normal return type,
while $C$ is the type of exceptions it could raise. For a concrete example of a
set $\funset$ of function types, refer to the resource polymorphic typing of
\code{sqdist} in the overview section \ref{sec:overview-aara}.

Note that we use a more readable notation for resource annotated types
throughout the overview section \ref{sec:overview-aara}. $L^q(\tau)$ is
shorthand for $\liT{\aty{\tau}{q}}$, $\tau_1^p + \tau_2^q$ is shorthand for
$\sumT{\aty{\tau_1}{p}}{\aty{\tau_2}{q}}$, and $\tau_1 \to^p_{q\odot q_e} \tau_2
  \odot \tau_e$ is shorthand for $\aty{\tau_1}{p} \to \aty{\tau_2}{q} \odot
  \aty{\tau_e}{q_e}$.

We also introduce the function $\abs{\cdot}$ that takes in a resource annotated
type and zeroes all the potential annotations within positive types. It is
defined in Figure \ref{fig:type-zeroing} and will be useful for defining various
typing rules in the rest of the paper. We extend it to contexts too, defining
$\abs{\Gamma}$ as the pointwise application of $\abs{\cdot}$ to each type in
$\Gamma$.

\begin{figure}[t]
  \framebox{$\abs{\tau}$ \quad $\abs{A}$} \hspace{.5em} zeroing \hfill $\;$
  \def \MathparLineskip {\lineskip=0.3cm}
  \begin{mathpar}
    \abs{\funset}                 = \funset

    \abs{\unitT}                  = \unitT

    \abs{\prodT{\tau_1}{\tau_2}}  = \prodT{\abs{\tau_1}}{\abs{\tau_2}}

    \abs{\voidT}                  = \voidT

    \abs{\sumT{A_1}{A_2}}         = \sumT{\abs{A_1}}{\abs{A_2}}

    \abs{\liT{A}}                 = \liT{\abs{A}}

    \abs{\aty{\tau}{q}}           = \aty{\abs{\tau}}{0}
  \end{mathpar}
  \vspace{-3ex}
  \caption{Zeroing Resource Annotated Types}
  \vspace{-3ex}
  \label{fig:type-zeroing}
\end{figure}

\subsection{Typing Rules for Values}

The value typing judgment \(\boxed{\Gamma; q \vdash v : \tau}\)
states that under context \(\Gamma\) and requiring \(q \in \Q_{\geq 0}\) additional
potential, value \(v\) has annotated type \(\tau\).
One way to understand these rules is that $\Gamma$ and $q$ provide the $\Phi(v:\tau)$ potential that $v$ carries.
Figure~\ref{fig:value-typing} contains the typing rules, in which the context
$\Gamma$ and the potential $q$ are treated linearly. We describe the key rules:

\begin{figure}[t]

  \framebox{$\Gamma; q \vdash v : \tau$} \hspace{.5em} value typing \hfill $\;$

  \def \MathparLineskip {\lineskip=0.3cm}
  \begin{mathpar}
    \Rule{T-Var}
    { }
    {x:\tau; 0 \vdash x:\tau }

    \Rule{T-Pair}
    { \Gamma_1 ; q_1 \vdash v_1 : \tau_1 \\ \Gamma_2; q_2 \vdash v_2 : \tau_2 }
    { \Gamma_1, \Gamma_2; q_1 + q_2 \vdash \pairS{v_1}{v_2} : \prodT{\tau_1}{\tau_2} }

    \Rule{T-Inl}
    { \Gamma; p \vdash v : \tau_1 }
    { \Gamma; q_1 + p \vdash \inlS{v} : \sumT{\aty{\tau_1}{q_1}}{A_2} }

    \Rule[rightskip=1em]{T-Fun}
    { \isZero{\Gamma}\!\!\!
      \\ \forall (\funT{\aty{\tau}{q}}{B}{C}) \in \funset.\ \Gamma; f : \funset, x : \tau; q \vdash e : \exty B C
    }
    { \Gamma; 0 \vdash \funS f x e : \funset}

    \Rule{T-Inr}
    { \Gamma; p \vdash v : \tau_2 }
    { \Gamma; q_2 + p \vdash \inrS{v} : \sumT{A_1}{\aty{\tau_2}{q_2}} }
    \\
    \Rule{T-Unit}
    { }
    { \cdot; 0 \vdash \trivS : \unitT }

    \Rule{T-Nil}
    { }
    { \cdot; 0 \vdash \nilS : \liT A}

    \Rule{T-cons}
    { \Gamma_1; q_1 \vdash v_1 : \tau
      \\ \Gamma_2; q_2 \vdash v_2 : \liT A
      \\ A = \aty \tau p
    }
    { \Gamma_1,\Gamma_2; q_1 + q_2 + p \vdash \consS{v_1}{v_2} : \liT A}
  \end{mathpar}
  \caption{Typing rules for values}
  \vspace{-3ex}
  \label{fig:value-typing}
\end{figure}

\begin{description}
  \item[T-Var] All the potential is provided by the context
        $x:\tau$ and $0$ additional potential is required.
  \item[T-Pair] The context $\Gamma_1,\Gamma_2$ and potential
        $q_1 + q_2$ are split, and $\Gamma_i; q_i$ is used to
        type $v_i$.
  \item[T-Inl/T-Inr] The potential $q_i + p$ is split. The context $\Gamma$ and potential $p$ are used to type $v$. $q_i$ additional potential is required depending on which injection it is.
  \item[T-Cons] The context $\Gamma_1,\Gamma_2$ and potential
        $q_1 + q_2 + p$ are split. Depending on $i$, $\Gamma_i; q_i$ types the head or tail of the list. $p$ additional potential is required, as each list element must carry $p$ potential.

  \item[T-Fun]
        No potential is consumed when typing a function. $q = 0$, and the
        $\isZero{\Gamma}$ premise ensures that the context carries no potential.
        This guarantees that the only source of potential when evaluating the
        function body is the function argument. This way, functions are
        potential-free and can be invoked any number of times.
        The rule derives a result type $\funset$ that contains possibly
        infinite arrow types $\funT{\aty{\tau}{q}}{B}{C}$, where $B$ is the
        annotated return type and $C$ is the annotated exception type,
        explained further in the subsequent type rules for computations.
        Each of these arrow types is justified by typing the function body $e$
        given the argument $x:\tau$, accompanying constant potential $q$, the
        set of function types $\funset$ for recursive calls, and the
        potential-free closure $\Gamma$.
\end{description}

The linear treatment of potential enables us to establish the following lemma
for closed values, proved by induction on the value typing judgment.
\begin{lemma}
  If $\emptyctx; q \vdash v : \tau$ and $\emptyctx; q' \vdash v : \tau$ then $q = q'$.
\end{lemma}
Since closed values must be typed with a unique $q$, we can define the potential
function $\Phi(\_ : \tau)$ that maps values of type $\tau$ to non-negative
rational numbers by setting $\Phi(v : \tau) = q$ if $\cdot; q \vdash v : \tau$.
We call $q$ the potential of $v$ under $\tau$.
With $\Phi$, we can state a useful lemma regarding the zeroing function
$\abs{\cdot}$, also proved by induction on the typing judgment.
It states that $\isZero{\tau}$ ensures that values of type $\tau$ must
carry no potential. We refer to such types as potential-free.
\begin{lemma}
  \label{lem:zero}
  If $\isZero{\tau}$ and $\Phi(v : \tau) = q$, then $q = 0$.
\end{lemma}

\subsection{Typing Rules for Computations and Exceptions}

The computation typing judgment \(\boxed{\Gamma; q \vdash e : \exty B C}\)
states that under context \(\Gamma\) and requiring \(q \in \Q_{\geq 0}\)
additional potential, computation \(e\) can either evaluate to a value and
potential of annotated type \(B\), or raise an exception value and potential of
annotated type \(C\). Figure~\ref{fig:comp-typing} contains the syntax-directed
computation typing rules, which all treat potential linearly. We describe some
key rules:

\begin{figure}[t]

  \framebox{$\Gamma; q \vdash e : \exty B C$} \hspace{.5em} computation typing \hfill $\;$

  \def \MathparLineskip {\lineskip=0.3cm}
  \begin{mathpar}
    \Rule[right=T-Val]{}
    { \Gamma; q \vdash v : \tau}
    { \Gamma; q + p \vdash \retS v  : \exty {\aty \tau p} C}

    \Rule[right=T-Let]{}
    { \Gamma_1; q \vdash e_1 : \exty {\aty \tau {q'}} C
      \\ \Gamma_2,x:\tau; q' \vdash e_2 : \exty B C
    }
    { \Gamma_1,\Gamma_2; q \vdash \letS{e_1}{x}{e_2}  : \exty B C}\\

    \Rule[right=T-Tick$^+$]{ }
    { }
    { \cdot; q + p \vdash \tickS q  : \exty {\aty \unitT p} C}\!\!

    \Rule[right=T-Tick$^-$]{ }
    { }
    { \cdot; p \vdash \tickS {-q}  : \exty {\aty \unitT {p + q}} C}\!\!

    \Rule[right=T-App]{ }
    { \Gamma_1; 0 \vdash v_1 : \funset
      \\ (\aty \tau {q_1} \to \exty B C) \in \funset
      \\ \Gamma_2; q_2 \vdash v_2 : \tau
    }
    { \Gamma_1,\Gamma_2; q_1 + q_2 \vdash \appS {v_1} {v_2}  : \exty B C}\\

    \Rule{T-Casevoid}
    {
      \Gamma; q \vdash v : \voidT\\
    }
    { \Gamma; q \vdash \casevoidS{v}  : \exty B C}

    \Rule{T-Casepair}
    {
      \Gamma_1; q_1 \vdash v : \prodT{\tau_1}{\tau_2} \\
      \Gamma_2, x_1 : \tau_1, x_2 : \tau_2 ; q_2 \vdash e : \exty B C
    }
    { \Gamma_1, \Gamma_2; q_1 + q_2 \vdash \casepairS{v}{x_1}{x_2}{e}  : \exty B C}\\

    \Rule{T-Casesum}
    {
      \Gamma_1; q_1 \vdash v : \sumT{\aty{\tau_1}{p_1}}{\aty{\tau_2}{p_2}}\\
      \Gamma_2, x_1 : \tau_1 ; p_1 + q_2 \vdash e_1 : \exty B C\\
      \Gamma_2, x_2 : \tau_2 ; p_2 + q_2 \vdash e_2 : \exty B C
    }
    { \Gamma_1, \Gamma_2; q_1 + q_2 \vdash \casesumS{v}{x_1}{e_1}{x_2}{e_2}  : \exty B C}\\

    \Rule[right=T-Caselist]{ }
    { \Gamma_1; q_1 \vdash v : \liT {\aty \tau p}\!\!\!\!
      \\ \Gamma_2; q_2 \vdash e_1 : \exty B C\!\!\!\!
      \\ \Gamma_2,x:\tau,y:\liT {\aty \tau p} ; q_2 + p \vdash e_2 : \exty B C
    }
    { \Gamma_1,\Gamma_2; q_1 + q_2 \vdash \caseS{v}{e_1}{x}{y}{e_2}  : \exty B C}\!\!\!\!
    \\

    \Rule[right=T-Raise]{}
    {  \Gamma; q \vdash v : \sigma }
    { \Gamma; q + r \vdash  \raiseS v  : \exty A {\aty \sigma r}}
    \!\!\!\!\!\!\!\!

    \Rule[right=T-Try]{}
    { \Gamma_1; q \vdash e_1 : \exty A {\aty \sigma r}
      \\ \Gamma_2, x:\sigma; r \vdash e_2 : \exty A C
    }
    { \Gamma_1,\Gamma_2; q \vdash \tryS{e_1}{x}{e_2} : \exty A C}
    \!\!\!\!\!\!\!\!\!\!\!\!
  \end{mathpar}
  \caption{Typing rules for computations}
  \vspace{-3ex}
  \label{fig:comp-typing}
\end{figure}

\begin{description}
  \item[T-Tick$^+$] Ticks with a positive argument consume resources. The $q +
          p$ input potential covers the {\sf tick} that consumes $q$ potential,
        which gets typed as $\unitT$ with $p$ potential remaining. The
        annotated exception type $C$ is arbitrary.
  \item[T-Tick$^-$] Ticks with a negative argument return resources. The $p$
        input potential is arbitrary, and the {\sf tick} is typed as
        $\unitT$ with $p + q$ potential attached, since $q$ potential is
        being returned. The annotated exception type $C$ is arbitrary.
  \item[T-Val] This rule lifts the typing judgment for value to a computation
        typing by threading through $p$ additional potential alongside the value
        as the constant return and an arbitrary type $C$ for the exception
        payload. This is justified since values cannot raise exceptions.
  \item[T-App] The context $\Gamma_1, \Gamma_2$ and potential $q_1 + q_2$ are
        both split. $\Gamma_1$ is used to type $v_1$, which being a function,
        requires $0$ additional potential. $\Gamma_2$ and $q_2$ are used to type
        $v_2$, the function argument. Finally, the resulting annotated result
        and exception types $\exty B C$ are valid if the arrow type $\aty \tau
          {q_1} \to \exty B C$ is in $\funset$.
  \item[T-Let] The context $\Gamma_1, \Gamma_2$ is split. $\Gamma_1; q$ types
        $e_1$. $e_1$'s annotated type's potential $q'$, along with $\Gamma_2$
        augmented with $x:\tau$ are used to type $e_2$, whose type $B$ is the
        overall result type. Most importantly for this paper, both $e_1$ and
        $e_2$ must agree on annotated exception type $C$ so that raising an
        exception in either computation would redirect control flow the same
        way.
  \item[T-Caselist] The context $\Gamma_1, \Gamma_2$ and potential $q_1 + q_2$ are both split. $\Gamma_1$ and $q_1$ are used to type $v$ as a list $\liT{\aty{\tau}{p}}$. The nil branch $e_1$ is typed using $\Gamma_2$ and $q_2$. The cons branch $e_2$ also uses $\Gamma_2$ and $q_2$, but augmented with bound variables $x$ for the head, typed $\tau$, and $y$ for the tail, typed the same as the original list, as well as $p$ additional potential made available from the head. As in \textsc{T-Casesum}, the branches must agree on their final $\exty B C$.
  \item[T-Raise] The potential $q + r$ is split. The context $\Gamma$ along with potential $q$ are used to type $v$, the exception payload, as some $\sigma$. The raise expression then has an annotated exception type of $\aty{\sigma}{r}$. Its annotated result type $A$ is arbitrary since it will never evaluate normally.
  \item[T-Try] The context $\Gamma_1, \Gamma_2$ are split. $\Gamma_1; q$ types
        the client expression $e_1$. The annotated exception type
        $\aty{\sigma}{r}$ of this typing is used when typing the ``handler''
        expression $e_2$, which is typed using $\Gamma_2$ augmented with bound
        variable $x : \sigma$ and potential $r$. Importantly, the annotated
        result type $A$ of both the main and handler expressions must agree.
        Furthermore, note that the handler expression itself may raise an
        exception too, and that its annotated exception type $C$ is the overall
        exception type.
\end{description}

Observe the symmetry between \textsc{T-Let} and \textsc{T-Try}. In \textsc{T-Let}, normal return composes between $e_1$ and $e_2$,
while exceptional flow $\odot C$ falls through between them. In \textsc{T-Try}, normal return falls through $e_1$ and $e_2$, while exceptional flow composes between.
Furthermore, as alluded to in the overview section \ref{sec:overview}, there is no globally enforced annotated exception type. Instead, only local contracts are enforced between points where an exception is raised and its  nearest enclosing handler (or function). This is crucial in improving the efficiency of the resource analysis; the same exception handled at two different places need not carry the same amount of potential.

\begin{figure}[t]
  \framebox{$\Gamma; q \vdash e : \exty B C$} \hspace{.5em} computation typing \hfill $\;$

  \begin{mathpar}
    \Rule{T-ContFun}
    { \Gamma, x_1 : \funset, x_2 : \funset; q \vdash e : \exty B C}
    { \Gamma, x : \funset; q \vdash [x,x / x_1,x_2]e : \exty B C }

    \Rule{T-WeakVar}
    { \Gamma; q \vdash e : \exty B C }
    { \Gamma,x:\tau; q \vdash e : \exty B C }

    \Rule{T-WeakPot}
    { \Gamma; q' \vdash e : \exty B C \\ q \ge q' }
    { \Gamma; q \vdash e : \exty B C }
  \end{mathpar}
  \caption{Structural typing rules for computations}
  \vspace{-3ex}
  \label{fig:comp-typing-structural}
\end{figure}

\subsection{Structural Typing Rules for Computations}
\label{sec:structural-rules}
Thus far, all the typing rules presented are syntax-directed and strictly
linear. Such a type system would mostly behave as we would like, with two main
issues to address. The obvious issue is that a fully linear system would require
that all execution paths of a program consume the same amount of resources, no
matter what branches are taken (in {\sf caselist}, {\sf casesum}, and {\sf
    try}). This is in conflict with the way we decorate programs with {\sf tick}
when evaluating metrics such as number of function calls, which typically varies
across different execution paths. Thus we make our type system {\it affine}
instead of linear by adding the {\sc T-WeakVar} and {\sc T-WeakPot} rules,
allowing variables and potential to be discarded at any point.

The subtler issue is that AARA functions are meant to be {\it structural}.
Recall that in our discussion of the {\sc T-Fun} rule, the premise
$\isZero{\Gamma}$ is only necessary because we expect the function to be
called multiple times, so it must not capture potential from its closure. We
add the typing rule {\sc T-ContFun} to enable structural functions to be
copied arbitrarily many times. This is also necessary for allowing a recursive
function to make multiple recursive calls.


A notable and deliberate difference between our presentation of the structural
rules and previous presentations of AARA is the lack of an explicit sharing
judgment. In previous presentations of AARA, the sharing judgment
$\share{\tau}{\tau_1}{\tau_2}$ relates three annotated types which are
structurally identical, differing only at their potential annotations. The
potential carried by a value $v$ of annotated type $\tau$ is the sum of the
potentials carried by the same value $v$ using annotations $\tau_1$ and $\tau_2$
instead. This judgment is then used to govern a contraction rule that preserves
potential:
\[
  \Rule[Right=T-Cont]{}
  { \share{\tau}{\tau_1}{\tau_2}\\ \Gamma, x_1 : \tau_1, x_2 : \tau_2; q \vdash e : \exty B C}
  { \Gamma, x : \tau; q \vdash [x,x / x_1,x_2]e : \exty B C }
\]
We argue that {\sc T-Cont} is unnecessary in our type system. A well-known fact
about linear type systems in general is that the contraction of purely positive
types is admissible, i.e. types like $\unitT + \unitT \times \unitT$ can be
copied without modifying the typing rules. A similar fact applies to our type
system\footnote{This is in fact true of the previous AARA type systems too.} --
purely positive types admit sharing. The only contraction rule we must add is
  {\sc T-ContFun}, which is needed to copy functions. The admissibility of sharing
is formally expressed in the following lemma.
\begin{lemma}
  \label{lem:admit-sharing}
  If $\tau, \tau_1, \tau_2$ are types such that the following statement holds:
  \begin{center}
    If $\Phi(v : \tau) = q$, then $q = \Phi(v : \tau_1)  + \Phi(v : \tau_2)$.
  \end{center}
  \qquad Then then we can construct $\curlyveedownarrow_{\tau,\tau_1,\tau_2} :
    \funT{\tau}{(\prodT{\tau_1}{\tau_2})}{\voidT}$, which encodes sharing.
\end{lemma}


The lemma can be proven by induction on the type $\tau$. The $\voidT$ exception
type ensures that $\curlyveedownarrow_{\tau,\tau_1,\tau_2}$ does not ``cheat''
by raising an exception to achieve the desired result type -- as there are no
closed values of type $\voidT$, it becomes impossible to raise an exception.
Compiler passes can be easily implemented to translate a fully structural
(contraction and weakening for all types) frontend language into our language,
appropriately inserting calls to $\curlyveedownarrow_{\tau,\tau_1,\tau_2}$, {\sc
    T-WeakVar}, and {\sc T-WeakPot}. Together, this allows values to be be reused or
discarded while potential is managed affinely.

\subsection{Typing Rules for Computational Effects}
It turns out that the typing rules for exceptions can be naturally extended to support effects and their handlers instead. We start by removing {\sf raise} and {\sf try} from the computation grammar, instead replacing them with the following:
\vspace{-1ex}
\[
  e ::= \ldots \mid \doS{\ell}{v} \mid \newhandlerS{e}{\ell}{x}{c}{e_\ell}{\Delta}{y}{e_r}
  \vspace{-1ex}
\]
The syntactic form \(\doS{\ell}{v}\) performs effect \(\ell\) with payload
\(v\). The syntactic form \(\newhandlerS{e}{\ell}{x}{c}{e_\ell}{\Delta}{y}{e_r}\) is an
effect handler that handles any effect \(\ell\) in effect signature \(\Delta\)
performed during the execution of computation \(e\). Effect \(\ell\) is handled
by branch \(e_\ell\), which may use payload \(x\) as well as continuation \(c\)
to resume the computation of \(e\) from the site where the effect was performed.
If \(e\) evaluates to a pure value, this value is substituted for \(y\) in the
return branch \(e_r\). Formally, effect signatures $\Delta$ have the following grammar:
\vspace{-1ex}
\[
  \Delta ::= \cdot \mid \Delta, \effectT{\ell}{A}{B}\
  \vspace{-1ex}
\]
Each effect \(\ell\) has a type specification \(A \Rightarrow B\), indicating
that it expects a payload of annotated type \(A\), and if resumed via the
continuation available to the handler, it will be replaced by a value of
annotated type \(B\). For example, \(\effectT{\sf inc}{\sf int}{\sf int}\) could
indicate an effect that increments some counter by the provided payload, and if
handled, return the updated value of the counter.

Furthermore, our typing judgments needs to be modified to accommodate effects
instead of exceptions. The value typing judgment is unchanged, but our new
judgment for computation typing is \(\boxed{\Gamma; q \vdash e : \exty B
  \Delta}\), where effect signature \(\Delta\) contains all the effects that
computation \(e\) can perform. Fortunately, the previous typing rules that do
not concern effects can all be retained by \(\alpha\)-renaming all \(C\)'s to
\(\Delta\)'s then reinterpreting all the annotated exception types as effect
signatures instead. This change similarly impacts the typing of functions.
Whereas the elements of the function set \(\funset\) had elements of the form
\(\funT{A}{B}{C}\), the elements now have the form \(\funT{A}{B}{\Delta}\).


\begin{figure}[t]

  \framebox{$\Gamma; q \vdash e : \exty B \Delta$} \hspace{.5em} computation typing \hfill $\;$

  \def \MathparLineskip {\lineskip=0.3cm}
  \begin{mathpar}
    \Rule{T-Linapp}
    { \Gamma_1; p \vdash v_1 : \linfunT{\aty{\tau}{q_1}}{B}{\Delta}
      \\ \Gamma_2; q_2 \vdash v_2 : \tau
    }
    { \Gamma_1,\Gamma_2; p + q_1 + q_2 \vdash \linappS {v_1} {v_2}  : \exty B \Delta}

    \Rule{T-Do}
    {  \Gamma; p \vdash v : \tau_1\\
      \effectT{\ell}{\aty{\tau_1}{q_1}}{\aty{\tau_2}{ q_2}} \in \Delta
    }
    { \Gamma; p + q_1 \vdash  \doS{\ell}{v}   : \exty {\aty{\tau_2}{q_2}} {\Delta}}

    \Rule[Right=T-Handle]{}
    { \Gamma_1; p \vdash e : \exty{\aty{\tau}{q}}{\Delta}\\
      \Gamma_3, y : \tau; q \vdash e_r : \exty{C}{\Delta'}\\
      \isZero{\Gamma_2}\\
      \forall \effectT{\ell}{\aty{\tau_1}{q_1}}{B} \in \Delta.\ (
      \Gamma_2, x : \tau_1, c : \linfunT{B}{C}{\Delta'}; q_1 \vdash
      e_\ell : \exty{C}{\Delta'})
    }
    { \Gamma_1, \Gamma_2, \Gamma_3; p \vdash \newhandlerS{e}{\ell}{x}{c}{e_\ell}{\Delta}{y}{e_r} : \exty{C}{\Delta'} }
  \end{mathpar}

  \framebox{$\Gamma; q \vdash v : \tau$} \hspace{.5em} value typing \hfill $\;$
  \vspace{-3ex}
  \begin{mathpar}
    \Rule[Right=T-LinFun]{}
    { \Gamma, x : \tau; p + q \vdash e : \exty B \Delta
    }
    { \Gamma; p \vdash \linfunS{x}{e} : \linfunT{\aty{\tau}{q}}{B}{\Delta}}
  \end{mathpar}
  \vspace{-3ex}
  \caption{New typing rules for effect handlers and linear functions}
  \vspace{-3ex}
  \label{fig:comp-typing-effects}
\end{figure}

\paragraph{Linear Functions}
To keep resource analysis tractable, we enforce that the continuations \(c\) in
the effect handler \(\newhandlerS{e}{\ell}{x}{c}{e_\ell}{\Delta}{y}{e_r}\) are
all one-shot by typing continuations as linear functions. This necessitates
adding linear functions to the language:
\vspace{-1ex}
\begin{mathpar}
  \tau ::= \hdots \mid \linfunT{A}{B}{\Delta}

  v ::= \ldots \mid \linfunS{x}{e}

  e ::= \ldots \mid \linappS{v_1}{v_2}
  \vspace{-1ex}
\end{mathpar}
With the addition of linear function types, we extend the potential zeroing
function $\abs{\cdot}$ (defined in Figure \ref{fig:type-zeroing}).
Interestingly, the right move is to make $\abs{\cdot}$ partial, and reject
linear function types as an input. This errs on the side of caution, ensuring
that $\isZero{\linfunT{A}{B}{\Delta}}$ never holds, meaning linear functions can
never get captured in the closures of structural functions, where they would be
at risk of being invoked multiple times (nonlinearly).

Additionally, whereas our structural function types are resource polymorphic
sets, we define linear function types to always be a single resource monomorphic
type. This restriction simplifies type inference and suffices for most use
cases, including all programs in our evaluation.

\paragraph{New Typing Rules}
The new typing rules for linear functions and computational effects are
presented in figure \ref{fig:comp-typing-effects}. We describe the key rules:
\begin{description}
  \item [T-LinFun] The linear function abstraction allows us to define functions
        that can invoke one-shot continuations available to effect handlers, as
        well as other linear functions, which is forbidden to structural
        functions \(\funS{f}{x}{e}\). This is reflected in its typing rule,
        where the lack of constraints on \(\Gamma\) and $p$ indicate that it may
        capture potential from its closure.
  \item[T-Do] Much like {\sc T-Raise}, the potential \(p + q_1\) is split. The
        context \(\Gamma\) along with potential \(p\) are used to type \(v\),
        the effect payload, as \(\tau_1\). \(\aty{\tau_1}{q_1}\) is the
        annotated input type of this effect \(\ell\) as specified in in the
        effect signature. The crucial difference is that computation can resume
        if an enclosing effect handler calls its continuation. Thus instead of
        having an arbitrary result type, it has result type
        \(\aty{\tau_2}{q_2}\), which is the annotated output type of effect
        \(\ell\) as specified in the effect signature.

  \item[T-Handle] Context \(\Gamma_1, \Gamma_2, \Gamma_3\) is split. $\Gamma_1$
        and $p$ are used to type \(e\), the effectful computation being handled.
        It has an effect signature \(\Delta\), which matches the branches of the
        handler. As for the handler branches, there are several crucial
        differences from the {\sc T-Try} rule.

        First, $\Gamma_2$, which is used to type each effect handler branch, has
        the constraint that $\isZero{\Gamma_2}$. This is because an effect
        handler branch can be executed an arbitrary number of times, just like
        the function body of a structural function, so we must forbid it from
        capturing potential or linear functions in \(\Gamma\). {\sc T-Try}
        lacked this restriction because exception handlers are executed at most
        once.

        Next, in addition to the effect's inputs payload \(x : \tau_1\) and
        potential \(q_1\), as specified by effect signature \(\Delta\), our
        handler branch is also able to use continuation \(c\). This continuation
        has a linear function type to enforce the one-shot discipline. Its
        argument type matches the output type of effect \(\ell\), since the
        eventual output of a {\sc do} is precisely the argument which we apply
        the continuation to. The continuation's result type \(C\) and its effect
        signature $\Delta'$ agree with the final result type, since the
        continuation reinstalls the effect handler. Given this context, each
        effect handler branch \(e_\ell\) must have type \(C\) and also agree on
        an effect signature \(\Delta'\).

        Finally, the return branch $e_r$ is typed using $\Gamma_3$ along with
        the pure result of $e$ bound to $y$. Its result type also agrees with
        the final result type $C$ and effect signature $\Delta'$.

\end{description}

\paragraph{Explicit Forwarding}
Note that our typing rule for effect handlers, {\sc T-Handle}, enforces explicit
effect forwarding \cite{DBLP:conf/icfp/KammarLO13} by requiring that a handler
handle all effects in $\Delta$, the effect signature of the computation $e$
being handled. The effects that are not morally handled still must be explicitly
handled, but their handler branches $e_l$ would forward the effect and thus
expose it within $\Delta'$. The result of enforcing this paradigm is that
whenever an effect is performed, it always gets handled by the nearest enclosing
handler (if one exists).

Concretely, a handler branch for an effect $\ell_f$ that
is explicitly forwarded could be implemented as
$$\{\ell_f : x.c. \letS{\doS{\ell_f}{x}}{y}{\linappS{c}{y}} \}$$ The handler
branch forwards the effect $\ell_f$, then also forwards the result of the effect
back to the computation being handled by applying the delimited continuation $c$
to the effect result $y$.

\paragraph{Subsuming Exceptions}
\label{sec:effect-exceptions}
We certainly hope to be able to implement exceptions using effects. Well, this
can be achieved by defining an exception effect whose input type is the
exception type \(C\) and output type is \(\voidT\). Exception handlers by design
cannot resume execution from where the exception was raised. Reusing the trick
we used for lemma \ref{lem:admit-sharing}, since there are no closed
values of type \(\voidT\), setting it to be the output type of the exception
effect enforces this restriction, as it becomes impossible for the effect
handler branch to construct an argument to call its continuation with.
Formally, here is the translation of exception types and the {\sf raise} and {\sf try} computations into our new effect handler setting, where $\ell_e$ is the exception effect label:
\vspace{-1ex}
\begin{align*}
  C                  & \rightsquigarrow \{\effectT{\ell_e}{C}{\aty{\voidT}{0}}\}                \\
  \raiseS{v}         & \rightsquigarrow \letS{\doS{\ell_e}{v}}{x}{\casevoidS{x}}                \\
  \tryS{e_1}{x}{e_2} & \rightsquigarrow \newhandlerS{e_1}{\ell_e}{x}{\_}{e_2}{\{\ell_e\}}{y}{y}
  \vspace{-1ex}
\end{align*}
Also, our translation of {\sf try} into {\sf handler} does not perfectly conserve potential, as {\sf try} allows the handler to use linear variables from the context, while {\sf handler} restricts the handler branches to structural variables. This restriction exists because our effect handlers are designed to handle resumable effects that can be handled multiple times, whereas exception handlers are only ever run once. To address this, we could add special cases for effect handler branches that never calls their continuations, but we omit that for brevity.

%% file: dynamics.tex
To define the cost semantics of our language, we introduce an abstract
machine called the K-machine~\cite{Harper12}. The K-machine
includes an explicit control stack that records the work that remains
to be done after a step is taken during the evaluation of a program.

A stack \(k\) is a list of stack frames, each of which is either a
binder \(x.e\) that sequences computation, or an effect handler awaiting the result of its computation \(\newhandlerS{\_}{\ell}{x}{c}{e_\ell}{\Delta}{y}{e_r}\). Having only two kinds of frames is a benefit that stems from the modal separation of values and computations.
Program evaluation is performed in a transition system with the following
kinds of machine states:
\begin{itemize}
  \item $k\eval e$. This says a closed expression $e$ is being evaluated on a \textit{stack} $k$.
  \item $k\return v$. This  says a closed value $v$ is being returned to a \textit{stack} $k$.
  \item $k\exnret  (\ell, v, k') $. This says an effect \(\ell\) is being propagated to stack \(k\), carrying payload \(v\) and accumulator for the delimited continuation \(k'\).
\end{itemize}
The full grammar is given below.
\[
  \begin{array}{r c l l}
    \text{Frame}~f & \Coloneqq
                   & x.e       & \text{sequence}                                                                        \\
                   & \mid      & \newhandlerS{\_}{\ell}{x}{c}{e_\ell}{\Delta}{y}{e_r}{\Delta} & \text{effect handler}   \\
    \text{Stack}~k & \Coloneqq
                   & \epsilon  & \text{empty stack}                                                                     \\
                   & \mid      & k;f                                                          & \text{frame sequence}   \\
    \text{State}~s & \Coloneqq
                   & k\eval e  & \text{evaluate expression}                                                             \\
                   & \mid      & k\return v                                                   & \text{return value}     \\
                   & \mid      & k \exnret (\ell, v, k')                                      & \text{propagate effect}
  \end{array}
\]
We also introduce a new value for delimited continuations: $\dcontS{k}$. It is a
dynamic value, meaning it never appears in source programs of our language, but
rather arises during the steps of our dynamics, to be defined below. It reifies
a control stack \(k\) into a delimited continuation of a linear function type.
When we apply it to some argument value $v$, we transition to $k \return v$, and
return the final value this steps to.





We can now define the dynamics of our language via transitions between machine
states. An initial state is one that evaluates an expression on an empty stack:
$\epsilon \eval e;q$. A final state either returns a value to an empty stack
$\epsilon \return v;q$ or propagates an unhandled effect to an empty stack
$\epsilon \exnret (\ell, v, k);q$. Figure~\ref{fig:state-trans} contains select
rules (see Appendix~\refappendix{sec:appendix-dynamics} for complete set) for
the transition judgment $\boxed{s; q \steps s'; q'}$, which states that machine
state \(s\) given \(q \in \Q_{\geq 0}\) resources transitions to machine state
\(s'\) with \(q' \in \Q_{\geq 0}\) resources remaining. It makes use of the $@$
operator that appends two stacks, which is defined inductively over the
structure of the second stack as follows:
\vspace{-1ex}
\begin{mathpar}
  k ~@~ \epsilon := k

  k ~@~ k'; f    := (k ~@~ k'); f
  \vspace{-1ex}
\end{mathpar}

\begin{figure}[t]
  \framebox{$s; q \steps s'; q'$} \hspace{.5em} state transitions \hfill $\;$
  \def \MathparLineskip {\lineskip=0.3cm}
  \begin{mathpar}
    \Rule{D-let}{ }{k\eval\letS{e_1}{x}{e_2};q\steps k;x.e_2\eval e_1;q}

    \Rule[rightskip=1em,leftskip=1em]{D-seq}{ }{k;x.e\return v;q\steps k\eval [v/x]e;q}

    \Rule{D-tick}{p\geq q}{k\eval\tickS{q};p\steps k\return \trivS;p-q}

    \Rule{D-Ret}{ }{k\eval\retS{v};q\steps k\return v;q}


    \Rule{D-Cons}{ }{k\eval\caseS{\consS{v_1}{v_2}}{e_0}{x}{y}{e_1};q\steps k \eval [v_1,v_2/x,y]e_1;q}

    \Rule{D-Dcont}{ }{k \eval \linappS{\dcontS{k'}}{v}; q \steps k ~@~ k' \return v; q}\!\!\!\!\!\!\!\!\!\!\!\!\!\!\!\!\!\!\!\!\!\!\!\!\!


    \Rule{D-Capture}{ }{k;x.e\exnret (\ell, v, k');q\steps k\exnret (\ell, v, x.e ~@~ k');q}

    \Rule{D-Do}{ }{k\eval\doS{\ell}{v};q\steps k\exnret (\ell, v, \epsilon);q}

    \Rule{D-normal}{ }{k;\newhandlerS{\_}{\ell}{x}{c}{e_\ell}{\Delta}{y}{e_r} \return v;q \steps  k\eval [v/y]e_r;q}

    \Rule[Right=D-handle]{}
    { \ell' \in \Delta\\ d = \dcontS{\newhandlerS{\_}{\ell}{x}{c}{e_\ell}{\Delta}{y}{e_r} ~@~ k'} }
    {k;\newhandlerS{\_}{\ell}{x}{c}{e_\ell}{\Delta}{y}{e_r}{\Delta}\exnret (\ell', v, k');q
    \steps k\eval [v, d/x, c]e_{\ell'};q}
  \end{mathpar}
  \caption{Select State Transitions}
  \vspace{-3ex}
  \label{fig:state-trans}
\end{figure}

The rules are mostly standard and the resource component $q$ of the
machines state is not modified except in the rule \textsc{D-tick}. We discuss that rule, as well as other key rules effects.
\begin{description}
  \item[D-Tick] requires that there are enough resources \(p\) supplied to cover the \(q\) resources in \(\tickS{q}\). In the resulting state, there are then $p-q$ resources.
  \item[D-Do] performs an effect and is similar to the previously presented \textsc{D-raise} for exceptions. The key difference is that rather than just a payload value \(v\), it also returns the effect label \(\ell\) and an empty delimited continuation \(\cdot\) that will capture more frames as the propagation proceeds.
  \item[D-Capture] shows the propagate effect state \(k \exnret (\ell, v, k') \) captures sequence frames and accumulates them into $k'$ as the effect propagates down the stack.
  \item[D-Handle] handles an effect being propagated. Our type system ensures that the effect's label \(\ell\) will be in the signature \(\Delta\) of the handler frame. Handling the effect means we switch from propagating an effect to evaluating the handler body \(e_{\ell'}\). During this, the effect payload \(v\) is substituted for \(x\), and more importantly, a delimited continuation value is substituted for \(c\). \textit{Crucially, this delimited continuation is \(k'\), which consists of all the stack frames starting at and including the handler frame itself, all the way to the point where the effect was performed.} This ``reinstallation'' makes it clear that we are working with deep handlers.
  \item[D-Dcont] states that when a delimited continuation \(\dcontS{k'}\) is applied to a value \(v\) during linear function application, we simply append the stack within, \(k'\), to our current stack \(k\), and return \(v\) to it. Ultimately, when the \(k'\) part of the stack finishes running, whatever value it returns will just be returned to the original stack \(k\) to proceed on.
\end{description}
%


%% file: soundness.tex
We now present the type soundness for our language via progress and
preservation theorems regarding an abstract stack machine.
This proof technique is novel and one of the main contributions of
this paper.
We start with several lemmas and additional judgments.

\paragraph{Substitution Lemma}

A key ingredient to the type soundness is the following substitution
lemma. It states that if we substitute a value $v_1$ for a variable
$x_1:\tau_1$ in a well-typed value (or computation) then we can type
the resulting value (or computation) but have to account for the
potential $q_1 = \Phi(v_1 : \tau_1)$ of the value.

\begin{lemma}[Substitution]
  \label{lem:subst}

  If $\cdot;q_1\vdash v_1:\tau_1$:
  \begin{enumerate}
    \item If $\Gamma,x_1:\tau_1;q_2\vdash v:\tau$, then $\Gamma;q_1+q_2\vdash [v_1/x_1]v:\tau$.
    \item If $\Gamma,x_1:\tau_1;q_2\vdash e:B\odot \Delta$, then $\Gamma;q_1+q_2\vdash
            [v_1/x_1]e:B\odot \Delta$.
  \end{enumerate}
\end{lemma}

The lemma is proved by induction on the judgments
$\Gamma,x_1:\tau_1;q_2\vdash v:\tau$ and
$\Gamma,x_1:\tau_1;q_2\vdash e:B \odot \Delta$ respectively.
The proof can be found in the Appendix~\refappendix{sec:appendix-soundness}.

\begin{figure}[t]
  \framebox{$\exty{A_1}{\Delta_1} \takes k \takes \exty{A_2}{\Delta_2}$} \hspace{.5em} stack typing \hfill $\;$
  \center
  \def \MathparLineskip {\lineskip=0.3cm}
  \begin{mathpar}
    \Rule[right=K-Emp]{}{ }{\exty B \Delta \takes \epsilon \takes B\odot \Delta}

    \Rule[right=K-Bnd]{}
    {\exty A {\Delta_1} \takes k \takes \exty B {\Delta_2}\\
      x : \tau; q \vdash e : \exty B {\Delta_2}}
    {\exty A {\Delta_1} \takes k; x.e \takes \exty {\aty \tau q} {\Delta_2}}

    \Rule[right=K-Handler]{}
    {\exty{A}{\Delta_1} \takes k \takes \exty B {\Delta_3}\\
      y : \tau; q \vdash e_r : B \odot \Delta_3\\\\
      \forall \effectT{\ell}{\aty{\tau_1}{q_1}}{C} \in \Delta_2.\
      x : \tau_1, c : \linfunT{C}{B}{\Delta_3}; q_1 \vdash
      e_\ell : \exty{B}{\Delta_3}\\
    }
    {\exty{A}{\Delta_1} \takes k; \newhandlerS{\_}{\ell}{x}{c}{e_\ell}{\Delta_2}{y}{e_r} \takes \exty{\aty \tau q}{\Delta_2}}
  \end{mathpar}
  \vspace{-3ex}
  \caption{Type rules for Stacks}
  \vspace{-3ex}
  \label{fig:stack-typing}
\end{figure}

\paragraph{Type Judgments for Stacks and States}
Figure \ref{fig:stack-typing} contains the rules for the stack typing judgment
$\boxed{\exty{A_1}{\Delta_1} \takes k \takes \exty{A_2}{\Delta_2}}$, which
states that stack $k$ is well-formed when it accepts either a value of annotated
type $A_2$, or an effect in signature $\Delta_2$, and returns either a value of
annotated type $A_1$, or an effect in signature $\Delta_1$. The stack typing
judgment finally allows us to define a typing rule for \(\dcontS{k}\):
\[
  \Rule[Right=T-Dcont]{}
  {\exty{A_1}{\Delta_1} \takes k \takes \exty{A_2}{\Delta_2}}
  {\Gamma; 0 \vdash \dcontS{k} : \linfunT{A_2}{A_1}{\Delta_1}}
\]
That is, a delimited continuation encoded by stack \(k\) is a linear function that takes type \(A_2\) that \(k\) takes, an returns type \(A_1\) with effect signature \(\Delta_1\) that \(k\) outputs.

Figure \ref{fig:state-typing} contains the rules for judgment judgment $\boxed{q
    \vdash s}$, which specify that execution state \(s\) is well-formed given \(q\)
potential.

The stack typing judgment allows us to state the following lemma,
which says that two stacks can only be appended if the output type and effect
signature of the second stack agrees with the input type and effect signature of
the first stack. Following the definition of $~@~$, the lemma can be proved by
induction over the struture of the first stack.
\begin{lemma}[Stack Append Typing]
  \label{lem:stack-append}
  $\exty{B_2}{\Delta_2} \takes k_1 ~@~ k_2 \takes \exty{B_1}{\Delta_1}$ if and only if there exists $B, \Delta$ such that $\exty{B_2}{\Delta_2} \takes k_1 \takes \exty{B}{\Delta}$ and $\exty{B}{\Delta} \takes k_2 \takes \exty{B_1}{\Delta_1}$
\end{lemma}

\begin{figure}[t]
  \framebox{$q \vdash s$} \hspace{.5em} state typing \hfill $\;$
  \center
  \def \MathparLineskip {\lineskip=0.3cm}
  \begin{mathpar}
    \Rule[right=ST-Exp]{}
    {
      \cdot;q \vdash e: \exty B \Delta\\
      \_ \takes k \takes \exty B {\Delta}
    }
    {q \vdash k \eval e}\!\!\!\!\!\!\!\!\!

    \Rule[right=ST-Val]{}{\cdot;q_1 \vdash v:\tau\\ \_ \takes k\takes \exty {\aty \tau {q_2} } \Delta}{q_1 + q_2 \vdash  k\return v}\!\!\!\!\!\!

    \Rule[right=ST-Eff]{}
    {
      \cdot;q_1 \vdash v:\tau \\
      \effectT{\ell}{\aty{\tau}{p}}{C} \in \Delta'\\
      \_ \takes k\takes \exty B {\Delta'}\\
      \exty{B}{\Delta'} \takes k' \takes \exty{C}{\Delta''}
    }
    {q_1 + p \vdash k\exnret (\ell, v, k')}
  \end{mathpar}
  \vspace{-3ex}
  \caption{Type rules for Execution States}
  \vspace{-3ex}
  \label{fig:state-typing}
\end{figure}

%


\paragraph{Type Soundness} With the well-formedness of execution
states defined, we are finally able to state the theorems of soundness: preservation and progress.

\begin{theorem}[Preservation]
  If $q\vdash s$ and $s;q\steps s_0;q_0$ then $q_0\vdash s_0$.
\end{theorem}

\begin{proof}
  Preservation is proved by case analysis on the state transition judgment $s;q\steps s_0;q_0$. We present one case
  illustrative of normal computation, and one case illustrative of effect handling. The complete proof can be found in
  the Appendix~\refappendix{sec:appendix-soundness}.
  \begin{description}
    \item[D-Cons] $k\eval \caseS{\consS{v_1}{v_2}}{e_0}{x}{y}{e_1};q\steps k\eval
            [v_1,v_2/x,y]e_1;q$\\
          By the assumption, $\_ \takes k\takes B\odot \Delta$,
          $\emptyctx;q\proves\caseS{\consS{v_1}{v_2}}{e_0}{x}{y}{e_1}:B\odot \Delta$. It suffices to show $\emptyctx;q\proves
            [v_1,v_2/x,y]e_1:B\odot \Delta$. Induct on $\emptyctx;q\proves\caseS{\consS{v_1}{v_2}}{e_0}{x}{y}{e_1}:B\odot \Delta$. Only {\sc T-caselist} and {\sc T-WeakPot} are possible. We omit the latter.

          By the premises of \textsc{T-caselist},
          $x:\tau,y:\liT{\aty{\tau}{p}};q_2+p\proves  e_0:B\odot \Delta$, where
          $\emptyctx;q_1\proves \consS{v_1}{v_2}:\liT{\aty{\tau}{p}}$, $q=q_1+q_2$. Invert $\emptyctx;q_1\proves \consS{v_1}{v_2}:\liT{\aty{\tau}{p}}$ to get
          $q_1=q_1'+q_2'+p$, $\emptyctx;q_1'\proves v_1:\tau$, $\emptyctx;q_2'\proves
            v_2:\liT{\aty{\tau}{p}}$. Then apply the substitution lemma \ref{lem:subst}, $\emptyctx;q_2+p+q_1'+q_2'\proves [v_1,v_2/x,y]e_1:B\odot
            \Delta$, which is
          $\emptyctx;q\proves[v_1,v_2/x,y]e_1:B\odot \Delta$.

    \item[D-Handle] $k;\newhandlerS{\_}{\ell}{x}{c}{e_\ell}{\Delta}{y}{e_r} \exnret (\ell', v, k');q \steps$ $k\eval [v, d/x, c]e_{\ell'};q$\\
          when $\ell \in \Delta$ and where $d = \dcontS{\newhandlerS{\_}{\ell}{x}{c}{e_\ell}{\Delta}{y}{e_r}  ~@~ k'}$.\\
          By the premises of {\sc ST-Eff} on the left, we have
          (1)\ $q=q_1+p$,
          (2)\ $\cdot;q_1 \vdash v:\tau$,
          (3)\ $\effectT{\ell'}{\aty{\tau}{p}}{C} \in \Delta$,
          (4)\ $\_ \takes k; \newhandlerS{\_}{\ell}{x}{c}{e_\ell}{\Delta}{y}{e_r} \takes \exty{B'}{\Delta}$, and
          (5)\ $\exty{B'}{\Delta} \takes k' \takes \exty{C}{\Delta''}$.
          Invert (4) using {\sc K-Handler} to get
          (6) \(\_ \takes k \takes \exty {B} {\Delta_3}\),
          (7) for all $\effectT{\ell}{\aty{\tau_1}{q_1}}{C}$ in $\Delta$,  $x : \tau_1, c : \linfunT{C}{B}{\Delta_3}; q_1 \vdash e_\ell : \exty{B}{\Delta_3}$, and
          (8) $y : \tau'; q' \vdash e_r : \exty{B}{\Delta_3}$ such that $B' = \aty{\tau'}{q'}$.
          Next, we use {\sc K-Emp}, {\sc K-Handler}, (7), (8), and lemma \ref{lem:stack-append} to construct the stack typing \(\exty{B}{\Delta_3} \takes \newhandlerS{\_}{\ell}{x}{c}{e_\ell}{\Delta}{y}{e_r} ~@~ k' \takes \exty{C}{\Delta''}\). Then by {\sc T-Dcont}, we have (9) $\cdot; 0 \vdash d : \linfunT{C}{B}{\Delta_3}$.

          Now, we can apply the substitution lemma for both $x$ and $c$ in $e_\ell$. For $[v/x]$, we have (2). For $[{\sf dcont}(\ldots)/c]$, we have (9). For $e_\ell$, we have (7). Putting it all together, we have
          \[
            \cdot ; p + q_1 + 0 \vdash [v, d/x, c]e' : \exty B {\Delta_3}
          \]
          Finally, using $q = q_1 + p$ from (1), along with \(\_ \takes k \takes \exty {B} {\Delta_3}\) from (6), we can apply {\sc ST-Exp} to type the overall right hand side of the $\steps$.
  \end{description}
  \vspace{-4ex}
\end{proof}
\vspace{-1ex}
\begin{theorem}[Progress]
  If $q\proves s$, then either $s$ final or $s;q\steps s';q'$.
\end{theorem}

\begin{proof}
  Progress is proved by case analysis on the state typing judgment $q \proves s$. We present the case for the propagate
  effect state. The complete proof can be found in the
  Appendix~\refappendix{sec:appendix-soundness}.
  \begin{description}
    \item[ST-Eff] $q \vdash k\exnret (\ell', v, k')$\\
          Induct on structure of k.
          \begin{itemize}
            \item $k=\emptystk$: $k \exnret (\ell', v, k')$ is final

            \item $k=k_1;x.e$: By \textsc{D-Capture}, $k_1;x.e\exnret (\ell', v, k');q \steps k_1 \exnret (\ell', v, x.e ~@~ k');q$

            \item $k=k_1;\newhandlerS{\_}{\ell}{x}{c}{e_\ell}{\Delta}{y}{e_r} $: Our type system guarantees that $\ell' \in \Delta$. Then by \textsc{D-handle}, $k_1;\newhandlerS{\_}{\ell}{x}{c}{e_\ell}{\Delta}{y}{e_r} \exnret (\ell', v, k');q \steps k_1 \eval [v, d/x, c]e_{\ell'};q$
                  where $d = \dcontS{\newhandlerS{\_}{\ell}{x}{c}{e_\ell}{\Delta}{y}{e_r}  ~@~ k'}$
          \end{itemize}
  \end{description}
  \vspace{-3ex}
\end{proof}
\vspace{-2ex}


%% file: evaluation.tex
In this section, we first describe our implementation of AARA with exceptions
and effects. We then evaluate the accuracy and performance of the implementation
with 21 representative benchmarks, while comparing it against previous
implementations of AARA. Both the source code of our implementation and the 21
benchmarks are available in the associated artifact
\cite{aara_effects_artifact}.

\subsection{Implementation}
Our implementation of AARA with exceptions and effects is built upon a novel
framework that is modeled after RaML~\cite{ramlWeb} and implemented in
OCaml.
We refer to it as \tool{} in this paper.
\tool{} has an SML frontend and makes some command-line calls to the MLton SML
compiler~\cite{Weeks06}.
Therefore, all benchmarks and code examples in this paper are
implemented in SML.
As mentioned in Section~\ref{sec:statics}, the analysis implemented
in \tool{} is more complete than the formal development in this
paper and supports many SML features, including products, sums, and
recursive types.
SML features like polymorphism, nested pattern matching, and modules are not
supported by our type theory or \tool{}; instead, we rely on the MLton compiler
to do away with them, such as MLton's monomorphisation and pattern flattening
passes.

The implementation of the new typing rules for exceptions and effects in \tool{}
follows from the algorithmic versions (see
Section~\ref{sec:inference}) of the typing rules in
Section~\ref{sec:statics}. These typing rules generate linear constraints between potential annotations, thus reducing type inference to a linear
optimization problem for the whole program. This optimization problem is then solved by the Coin-Or LP solver~\cite{CoinOrCLP}.
%

\vspace{-1ex}
\paragraph{Exceptions Frontend}
SML has native support for raising exceptions via \code{{\bf raise} ExName} and handling them via \code{e {\bf handle} e'}. We directly use these syntactic forms when writing programs for \tool{}. Then to add exception support to \tool{}, we simply translate SML exceptions and handlers into their corresponding AST variants in the IR.

Since MLton performs whole program compilation, we combine all
exceptions into one large sum type, which serves as the payload of all
exceptions. For example, our running example from
Figure~\ref{fig:dist}, which declares two list length mismatch
exceptions, will use the \code{{\bf datatype} exn = Emis1 {\bf of} unit | Emis2 {\bf of} unit} for execptions.
As discussed in
\hyperref[sec:effect-exceptions]{Section~\ref{sec:effect-exceptions}},
another option is to translated exceptions into effects. However, we
retain both syntactic forms due to the restriction on effect handler
branches from using linear variables.

\vspace{-1ex}
\paragraph{Effects frontend}
Adding effect support was more challenging because SML does
not have native support for effects and their handlers.
Users define the effects they wish to perform in a datatype.
They can then perform effects using an SML function we supply called
\code{e\_do}, and handle effects using \code{e\_handle}. These
functions are translated into explicit {\sf do} and {\sf handler} AST
variants in the IR.

\vspace{-1ex}
\paragraph{Cost Metric}
As stated in the overview, the cost metric we chose in this paper is the
number of function calls and handler evaluations. To streamline the evaluation,
we added the command line flags \code{-capps} and \code{-chandles} to \tool{},
which respectively insert \code{(R.tick 1)} next to function calls and handlers
before running the main analysis pass. This enables us to omit \code{R.tick}
expressions from our benchmark programs and minimize inconsistencies in how we
evaluate benchmarks.


\subsection{Goals and Methodology}

The goal of our evaluation is to determine the efficiency of the new
analysis, the quality of the derived bounds, as well as the
overhead with respect to RaML, an AARA implementation that supports raising
exceptions but not handling them.

To this end, we have constructed 21 benchmark programs which employ the raising
and handling of exceptions and effects to induce nontrivial (non-local) control
flow.
Given these benchmarks, we automatically derive bounds using \tool{}
with exceptions/effects and compare them to manually derived bounds to
determine their tightness.
For all but two benchmarks, the cost metric chosen measures the number of
function calls and handler evaluations.
Additionally, we also test our implementation on 6 examples without exception
handlers to determine whether \tool{} has a significant performance overhead
when analyzing the same programs as RaML.
Since the extension of AARA to exceptions is conservative, we do not
expect to derive bounds that differ from the RaML bounds for these
benchmarks (and the experiments confirm this).

For each benchmark, we manually analyze the benchmarks by hand to determine
whether the automatically derived bounds by \tool{} are tight. Furthermore, for
the exception benchmarks, we also measure the time it takes for \tool{} to
perform its analysis. We run 10 trials for each benchmark and report the average
elapsed time across them. These experiments were run on a machine with an AMD
Ryzen 9 7940HS processor and 64 GB of RAM.

\begin{table}[t]
  \resizebox{\textwidth}{!}{
    \begin{tabular}{l|l|l|l|l|l}
                                                                                & \multicolumn{2}{l|}{RaML}                  & \multicolumn{3}{l}{\tool{} with Exceptions}                                                                                                                                                                                                                  \\ \hline
      Function                                                                  & Derived Bound                              & Time (s)                                    & Derived Bound                              & Tight?                                                                                                                                                 & Time (s) \\ \hline
      \code{sqdist(v1 : L(\float), v2 : L(\float)) : \float}                    & \(|\code{v2}|\)                            & 0.05                                        & \(|\code{v1}|\)                            & \checkmark                                                                                                                                             & 0.007    \\
      \code{distances\_1(vs : L(L(\float)), p : L(\float)) : L(\float)}         & \(1 + 3|\code{vs}| + \sum_i |\code{v}_i|\) & 0.07                                        & \(1 + 3|\code{vs}| + \sum_i |\code{v}_i|\) & \checkmark                                                                                                                                             & 0.02     \\
      \code{distances\_2(vs : L(L(\float)), p : L(\float)) : L(\float)}         & Unsolvable                                 & --                                          & \(1 + 3|\code{vs}| + \sum_i |\code{v}_i|\) & \checkmark                                                                                                                                             & 0.03     \\ \hline
      \code{nearest(vs : L(L(\float)), p : L(\float)) : (\float, L(L(\float)))} & Unsolvable                                 & --                                          & \(1 + 3|\code{vs}| + \sum_i |\code{v}_i|\) & \checkmark                                                                                                                                             & 0.12     \\
      \code{zip(l1 : L(a), l2 : L(b)) : L((a, b))}                              & Unsolvable                                 & --                                          & \(1 + 2|\code{l2}|\)                       & \checkmark                                                                                                                                             & 0.01     \\
      \code{map2(f : (a, b) $\to$ c, l1 : L(a), l2 : L(b)) : c}                 & Unsolvable                                 & --                                          & \(1 + 4|\code{l2}|\)                       & \checkmark                                                                                                                                             & 0.01     \\
      \code{eval\_exp(e : exp) : \integer}                                      & Unsolvable                                 & --                                          & \(1 + 2c + 2p + t\)                        & \checkmark                                                                                                                                             & 0.13     \\
      \code{sort(l : L(\integer)) : L(\integer)}                                & Unsolvable                                 & --                                          & \(|\code{l}| + \binom{|\code{l}|}{2}\)     & \checkmark \tablefootnote{This bound is tight because we use quicksort, which is worst case quadratic time. \tool{} cannot derive logarithmic bounds.} & 0.07     \\ \hline
      \code{zip'(l1 : L(a), l2 : L(b)) : L((a, b))}                             & \(|\code{l1}|\)                            & 0.03                                        & \(|\code{l1}|\)                            & \checkmark                                                                                                                                             & 0.01     \\
      \code{map2'(f : (a, b) $\to$ c, l1 : L(a), l2 : L(b)) : c}                & \(2|\code{l1}|\)                           & 0.03                                        & \(2|\code{l1}|\)                           & \checkmark                                                                                                                                             & 0.01     \\
      \code{eval\_exp'(e : exp) : \integer}                                     & \(2p + 2t\)                                & 0.04                                        & \(c + p + t\)                              & \checkmark                                                                                                                                             & 0.04     \\
      \code{sort'(l : L(\integer)) : L(\integer)}                               & \(\binom{|\code{l}|}{2}\)                  & 0.09                                        & \(\binom{|\code{l}|}{2}\)                  & \checkmark                                                                                                                                             & 0.06
    \end{tabular}
  }\\[2ex]
  \caption{\tool{} with Exceptions Evaluation Results}
  \vspace{-5ex}
  \label{table:evaluation}
\end{table}

\subsection{Exception Benchmarks}
Table~\ref{table:evaluation} contains the results of the evaluation of
exception benchmarks.
The first group of functions we evaluate in are the running examples covered in
the overview Section~\ref{sec:overview-exn}. The bounds for \code{distances\_1}
and \code{distances\_2} differ from the bounds derived in the overview because
we now count all function calls instead of calls to \code{sqdist} only.
In the second group, we evaluate additional functions that
contain both exceptions and handlers:
\begin{description}
  \item[nearest] The function \code{nearest} builds on the examples in
        Section~\ref{sec:overview-exn} and returns the shortest distance
        between any vector in \code{vs} and vector \code{p}, as well as the
        subset of \code{vs} satisfying this distance.  An interesting aspect
        of \code{nearest} is that it leverages exceptions with payloads to
        improve efficiency when the shortest distance is 0.

  \item[zip] The \code{zip} function combines two lists. However, if the two
        lists are not the same length, it throws an exception
        \code{UnequalLengths of ('a * 'b) list} whose payload is the
        partialliy zipped result. The function
        needs to repeatedly handle and re-raise the exception with a larger
        payload.

  \item[map2] Similarly, the function \code{map2} applies a function
        that takes two inputs across two lists, once again raising the same
        exception if they are not the same length. 

  \item[eval\_exp] This function is an evaluator for a simple expression
        datatype.
        It uses exceptions and handlers to shortcut multiplication
        whenever there is a 0. 

  \item[sort] The function \code{sort} does a linear pass over the
        input integer list to check that it is sorted. If it is not,
        an exception is thrown, and the surrounding handler calls
        quicksort to sort the list.
\end{description}

Finally, the last group of benchmarks are just functions from the previous group
with their exception handlers removed. For example, \code{eval\_exp'} no longer
uses exceptions to shortcut multiplication with 0, and \code{sort'} calls
quicksort without first performing the sortedness check in \code{sort}.

\begin{table}[t]
  \resizebox{0.65\textwidth}{!}{
    \begin{tabular}{l|l|l}
                                                                  & \multicolumn{2}{l}{\tool{} with Effects}                     \\ \hline
      Function                                                    & Derived Bound                                   & Tight?     \\ \hline
      \code{simple\_IO(() : \unitT) : \unitT}                     & 11                                              & \checkmark \\\hline
      \code{store\_lists(vs : L(L(\integer))) : \unitT}           & \(6 + 9|\code{vs}| + \sum_i |\code{v}_i|\)      & \checkmark \\
      \code{store\_lists\_q\_naive(vs : L(L(\integer))) : \unitT} & Unsolvable                                      & --         \\
      \code{store\_lists\_q(vs : L(L(\integer))) : \unitT}        & \(8 + 13|\code{vs}| + \sum_i |\code{v}_i|\)     & \checkmark \\
      \code{store\_suffix\_lists(vs : L(\integer)) : \unitT}      & \(14 + 10|\code{vs}| + \binom{|\code{vs}|}{2}\) & \checkmark \\\hline
      \code{generator\_to\_list(n : nat) : L(nat)}                & $1 + 3|\code{n}|$                               & \checkmark \\
      \code{generator\_to\_list\_2(n : nat) : L(nat)}             & $1 + 5|\code{n}| + \binom{|\code{n}|}{2}$       & \checkmark \\
      \code{generator\_counter(n : nat) : \unitT}                 & $3 + 11|\code{n}| + \binom{|\code{n}|}{2}$      & \checkmark \\
      \code{prefix\_sum\_list(ns : L(nat)) : L(nat)}              & $2 + 12|\code{ns}| + \sum_i |\code{ns}_i|$      & \checkmark
    \end{tabular}
  }\\[2ex]
  \caption{\tool{} with Effects Evaluation Results}
  \vspace{-5ex}
  \label{table:evaluation2}
\end{table}

\subsection{Effect Benchmarks}
The benchmarks for effect handlers are listed in Table~\ref{table:evaluation2}.
Since we are only evaluating \tool{} with Effects and not the original RaML
(which cannot analyze programs with effects), we did not measure the analysis
times for this benchmarks, as there would be nothing to compare against. The
first benchmark, \code{simple\_IO}, performs 3 IO effects. The second group
contains functions that use the same \code{Insert} and \code{Remove} effects as
our running example in the overview in Figure \ref{fig:store}:
\begin{description}
  \item[store\_lists\_q\_naive] is \code{store\_lists} (the overview example)
        with a modification to the handler. While the effect handler in
        \code{store\_lists} stores and removes elements as if it were a stack,
        \code{store\_lists\_q\_naive} implements a queue instead. As \code{naive} in
        the name suggests, this queue is naively implemented as a single list that
        requires either \code{Insert} or \code{Remove} to traverse it, incurring cost
        proportional to the size of the store. This cost is impossible to capture in
        the effect signature, thus this example is intentionally impossible to
        analyze.

  \item[store\_lists\_q] is an efficient implementation of \code{store\_lists\_q\_naive}. It uses the standard functional queue implementation of two lists, which has amortized constant cost for both \code{Insert} and \code{Remove} operations, which is now possible to type.

  \item[store\_suffix\_lists] is a twist on the original \code{store\_lists}. Its argument \code{vs} is no longer a list of lists of integers, but rather just a list of integers. We store all the suffixes of this list into the store. We then remove all of them from the store and traverse each one, just as in \code{store\_lists}. 
\end{description}

Finally, we evaluate four benchmarks that use the ${\sf Yield}: {\sf nat}
  \Rightarrow \unitT$ effect that comes from generators in the literature. They
all rely on \code{datatype nat = Zero | Succ of nat}. This representation allows
us to store potential in a \code{nat} proportional to the number it represents.
\begin{description}
  \item[generator\_to\_list] takes a natural number $n$, creates a generator
        that performs \code{Yield} on every natural number from $n$ to $0$, then uses
        a handler to collect the yielded numbers into a list.
  \item[generator\_to\_list\_2] is similar to \code{generator\_to\_list} but calls \code{traverse} (but for \code{nat}s) on each element of the output list. For reasons similar to \code{store\_suffix\_lists}, it also requires quadratic potential.
  \item[generator\_counter] creates the same generator as the previous two benchmarks, then uses the \code{Get} and \code{Put} effects from state to sum up the yielded numbers. It shows that \tool{} with Effects can analyze code that composes effect handlers.
  \item[prefix\_sum\_list] relies on a \code{prefix\_sum\_generator}, which uses
        the state effects \code{Get} and \code{Put} to keep track of the prefix
        sum of a list, and also performs the \code{Yield} effect at each prefix.
        On the other end, \code{prefix\_sum\_list} collects the yielded numbers
        back into a list. It shows that \tool{} with Effects can analyze code
        that performs multiple effects together (the state and generator
        effects), as well as the explicit forwarding of the \code{Yield} effect
        to the outer handler.
\end{description}


\subsection{Results}

Our evaluation results are summarized in
and Table~\ref{table:evaluation} and Table~\ref{table:evaluation2}.
It contains, for each benchmark, the function name and its type, the bound
derived by \tool{}, and its tightness. For the exception benchmarks, we also
provide the bounds derived by RaML, as well as analysis times for both RaML and
\tool{} (in the {\it Time (s)} columns). As indicated by the
checkmarks in the {\it Tight?} columns, all derived bounds are not only
asymptotically tight but have optimal constant factors too.

There are a few trivial differences between the bounds derived by RaML and
\tool{}. For \code{sqdist}, RaML derives \(|\code{v2}|\) while \tool{} derives
\(|\code{v1}|\). These bounds differ, but as explained in section
\ref{sec:overview-aara} of the overview, they are both valid and tight upper
bounds of the true cost $\min(|\code{v1}|,|\code{v2}|)$ that our type system
cannot directly encode.
In the bounds for the functions \code{eval\_exp} and
\code{eval\_exp'}, $c$, $p$, and $t$ respectively denote the number of
\code{Const}, \code{Plus}, and \code{Times} nodes.
For \code{eval\_exp'}, RaML derives $2p + 2t$ while \tool{} derives $c + p + t$.
While these bounds appear quite different at a glance, they are nearly equal,
since the number of leaves, $c$, in a binary tree, is always equal to the number
of interior nodes, $p + t$, plus one. We are not entirely sure why these two
tools report different bounds.

Our findings are that \tool{} with exceptions and effects performs better than
RaML on our benchmarks. It can analyze a larger set of user functions that raise
and handle exceptions and effects, while RaML can only analyze
functions that raise exceptions. Furthermore, by observing the time it takes to
analyze functions that both implementations can analyze, we see that despite the
additional complexity, our implementation does not add overhead. In fact, our
implementation is always faster, although that can be attributed to RaML
tracking multivariate potential functions (such as \(|\code{l1}||\code{l2}|\)),
which \tool{} does not support.

%% file: appendix.tex
\begin{appendices}
  \section{Type System for Linear-Bound AARA with Effects}
  \label{sec:appendix-type-system}

  \subsection{Syntax Grammar}
  \begin{minipage}{0.45\linewidth}
    \[
      \begin{array}{lll}
        v & ::= & x                \\
          &     & \funS f x e      \\
          &     & \trivS           \\
          &     & \pairS{v_1}{v_2} \\
          &     & \inlS{v}         \\
          &     & \inrS{v}         \\
          &     & \nilS            \\
          &     & \consS{v_1}{v_2} \\
          &     & \linfunS{x}{e}   \\
          &     & \dcontS{k}
      \end{array}
    \]
  \end{minipage}
  \begin{minipage}{0.45\linewidth}
    \[
      \begin{array}{lll}
        e & ::= & \retS v                                             \\
          &     & \tickS q                                            \\
          &     & \letS{e_1}{x}{e_2}                                  \\
          &     & \appS {v_1} {v_2}                                   \\
          &     & \casepairS{v}{x_1}{x_2}{e}                          \\
          &     & \casevoidS{v}                                       \\
          &     & \casesumS{v}{x_1}{e_1}{x_2}{e_2}                    \\
          &     & \caseS{v}{e_1}{x_1}{x_2}{e_2}                       \\
          &     & \doS{\ell}{v}                                       \\
          &     & \newhandlerS{e}{\ell}{x}{c}{e_\ell}{\Delta}{y}{e_r} \\
          &     & \linappS {v_1} {v_2}
      \end{array}
    \]
  \end{minipage}

  \subsection{Resource Annotated Types and Effect Signatures}
  \begin{mathpar}
    \tau ::= \funset \mid \unitT \mid \prodT{\tau_1}{\tau_2} \mid \voidT \mid \sumT{A_1}{A_2} \mid \liT A \mid \linfunT{A}{B}{\Delta}

    A, B, C ::= \aty \tau q

    \funset ::= \{\funT{A}{B}{\Delta} \mid \Theta \}

    \Delta ::= \cdot \mid \Delta, \effectT{\ell}{A}{B}
  \end{mathpar}

  \subsection{Value Typing}
  \framebox{$\Gamma; q \vdash v : \tau$}
  \begin{mathpar}
    \Rule{T-Var}
    { }
    {x:\tau; 0 \vdash x:\tau }

    \Rule{T-Fun}
    { \Gamma=\abs{\Gamma}
      \\ \forall (\funT{\aty{\tau}{q}}{B}{\Delta}) \in \funset.\ \Gamma; f : \funset, x : \tau; q
      \vdash e : \exty B \Delta
    }
    { \Gamma; 0 \vdash \funS f x e : \funset}\\

    \Rule{T-Unit}
    { }
    { \cdot; 0 \vdash \trivS : \unitT }

    \Rule{T-Pair}
    { \Gamma_1 ; q_1 \vdash v_1 : \tau_1 \\ \Gamma_2; q_2 \vdash v_2 : \tau_2 }
    { \Gamma_1, \Gamma_2; q_1 + q_2 \vdash \pairS{v_1}{v_2} : \prodT{\tau_1}{\tau_2} }\\

    \Rule{T-Inl}
    { \Gamma; p \vdash v : \tau_1 }
    { \Gamma; q_1 + p \vdash \inlS{v} : \sumT{\aty{\tau_1}{q_1}}{A_2} }

    \Rule{T-Inr}
    { \Gamma; p \vdash v : \tau_2 }
    { \Gamma; q_2 + p \vdash \inrS{v} : \sumT{A_1}{\aty{\tau_2}{q_2}} }\\

    \Rule{T-Nil}
    { }
    { \cdot; 0 \vdash \nilS : \liT A}

    \Rule{T-cons}
    { \Gamma_1; q_1 \vdash v_1 : \tau
      \\ \Gamma_2; q_2 \vdash v_2 : \liT A
      \\ A = \aty \tau p
    }
    { \Gamma_1,\Gamma_2; q_1 + q_2 + p \vdash \consS{v_1}{v_2} : \liT A}

    \Rule{T-LinFun}
    { \Gamma, x : \tau; p + q \vdash e : \exty B \Delta
    }
    { \Gamma; p \vdash \linfunS{x}{e} : \linfunT{\aty{\tau}{q}}{B}{\Delta}}

    \Rule{T-Dcont}
    {\exty{A_1}{\Delta_1} \takes k \takes \exty{A_2}{\Delta_2}}
    {\Gamma; 0 \vdash \dcontS{k} : \linfunT{A_2}{A_1}{\Delta_1}}
  \end{mathpar}

  \subsection{Computation Typing}
  \framebox{$\Gamma; q \vdash e : \exty B \Delta$}
  \begin{mathpar}
    \Rule[right=T-Val]{}
    { \Gamma; q \vdash v : \tau}
    { \Gamma; q + p \vdash \retS v  : \exty {\aty \tau p} \Delta}

    \Rule[right=T-Let]{}
    { \Gamma_1; q \vdash e_1 : \exty {\aty \tau {q'}} C
      \\ \Gamma_2,x:\tau; q' \vdash e_2 : \exty B C
    }
    { \Gamma_1,\Gamma_2; q \vdash \letS{e_1}{x}{e_2}  : \exty B \Delta}\\

    \Rule[right=T-Tick$^+$]{ }
    { }
    { \cdot; q + p \vdash \tickS q  : \exty {\aty \unitT p} \Delta}

    \Rule[right=T-Tick$^-$]{ }
    { }
    { \cdot; p \vdash \tickS {-q}  : \exty {\aty \unitT {p + q}} \Delta}

    \Rule{T-App}
    { \Gamma_1; 0 \vdash v_1 : \funset
      \\ (\aty \tau {q_1} \to \exty B \Delta) \in \funset
      \\ \Gamma_2; q_2 \vdash v_2 : \tau
    }
    { \Gamma_1,\Gamma_2; q_1 + q_2 \vdash \appS {v_1} {v_2}  : \exty B \Delta}

    \Rule{T-Casepair}
    {
      \Gamma_1; q_1 \vdash v : \prodT{\tau_1}{\tau_2}\\
      \Gamma_2, x_1 : \tau_1, x_2 : \tau_2 ; q_2 \vdash e : \exty B \Delta
    }
    { \Gamma_1, \Gamma_2; q_1 + q_2 \vdash \casepairS{v}{x_1}{x_2}{e}  : \exty B \Delta}\\

    \Rule{T-Casevoid}
    {
      \Gamma; q \vdash v : \voidT\\
    }
    { \Gamma; q \vdash \casevoidS{v}  : \exty B \Delta}

    \Rule{T-Casesum}
    {
      \Gamma_1; q_1 \vdash v : \sumT{\aty{\tau_1}{p_1}}{\aty{\tau_2}{p_2}}\\
      \Gamma_2, x_1 : \tau_1 ; p_1 + q_2 \vdash e_1 : \exty B \Delta\\
      \Gamma_2, x_2 : \tau_2 ; p_2 + q_2 \vdash e_2 : \exty B \Delta\\
    }
    { \Gamma_1, \Gamma_2; q_1 + q_2 \vdash \casesumS{v}{x_1}{e_1}{x_2}{e_2}  : \exty B \Delta}\\

    \Rule{T-Caselist}
    { \Gamma_1; q_1 \vdash v : \liT {\aty \tau p}
      \\ \Gamma_2; q_2 \vdash e_1 : \exty B \Delta
      \\ \Gamma_2,x:\tau,y:\liT {\aty \tau p} ; q_2 + p \vdash e_2 : \exty B \Delta
    }
    { \Gamma_1,\Gamma_2; q_1 + q_2 \vdash \caseS{v}{e_1}{x}{y}{e_2}  : \exty B \Delta}

    \Rule{T-Linapp}
    { \Gamma_1; p \vdash v_1 : \linfunT{\aty{\tau}{q_1}}{B}{\Delta}
      \\ \Gamma_2; q_2 \vdash v_2 : \tau
    }
    { \Gamma_1,\Gamma_2; p + q_1 + q_2 \vdash \linappS {v_1} {v_2}  : \exty B \Delta}

    \Rule{T-Do}
    {  \Gamma; p \vdash v : \tau_1\\
      \effectT{\ell}{\aty{\tau_1}{q_1}}{\aty{\tau_2}{ q_2}} \in \Delta
    }
    { \Gamma; p + q_1 \vdash  \doS{\ell}{v}   : \exty {\aty{\tau_2}{q_2}} {\Delta}}

    \Rule{T-Handle}
    { \Gamma_1; p \vdash e : \exty{\aty{\tau}{q}}{\Delta}\\
      \Gamma_3, y : \tau; q \vdash e_r : \exty{C}{\Delta'}\\
      \isZero{\Gamma_2}\\
      \forall \effectT{\ell}{\aty{\tau_1}{q_1}}{B} \in \Delta.\ (
      \Gamma_2, x : \tau_1, c : \linfunT{B}{C}{\Delta'}; q_1 \vdash
      e_\ell : \exty{C}{\Delta'})
    }
    { \Gamma_1, \Gamma_2, \Gamma_3; p \vdash \newhandlerS{e}{\ell}{x}{c}{e_\ell}{\Delta}{y}{e_r} : \exty{C}{\Delta'} }
  \end{mathpar}

  \subsection{Structural Typing Rules}
  \framebox{$\Gamma; q \vdash e : \exty B \Delta$}
  \begin{mathpar}
    \Rule{T-ContFun}
    { \Gamma, x_1 : \funset, x_2 : \funset; q \vdash e : \exty B \Delta}
    { \Gamma, x : \funset; q \vdash [x,x / x_1,x_2]e : \exty B \Delta }

    \Rule{T-WeakFun}
    { \Gamma; q \vdash e : \exty B \Delta }
    { \Gamma,x:\funset; q \vdash e : \exty B \Delta }

    \Rule{T-WeakPot}
    { \Gamma; q' \vdash e : \exty B \Delta \\ q \ge q' }
    { \Gamma; q \vdash e : \exty B \Delta }



  \end{mathpar}

  \newpage
  \section{Full Cost Semantics}
  \label{sec:appendix-dynamics}

  \subsection{State Transitions}
  \framebox{$s; q \steps s'; q'$}
  \begin{mathpar}
    \Rule{D-tick}{p\geq q}{k\eval\tickS{q};p\steps k\return \trivS;p-q}

    \Rule{D-ret}{ }{k\eval\retS{v};q\steps k\return v;q}

    \Rule{D-fun}{ }{k\eval\appS{\funS{f}{x}{e}}{v_2};q\steps k\eval [\funS{f}{x}{e},v_2/f,x]e;q}

    \Rule{D-let}{ }{k\eval\letS{e_1}{x}{e_2};q\steps k;x.e_2\eval e_1;q}

    \Rule{D-seq}{ }{k;x.e\return v;q\steps k\eval [v/x]e;q}

    \Rule{D-pair}{ }{k\eval\casepairS{\pairS{v_1}{v_2}}{x_1}{x_2}{e};q\steps k\eval [v_1,v_2/x_1,x_2]e;q}

    \Rule{D-inl}{ }{k\eval\casesumS{\inlS{v}}{x_1}{e_1}{x_2}{e_2};q\steps k\eval [v/x_1]e_1;q}

    \Rule{D-inr}{ }{k\eval\casesumS{\inrS{v}}{x_1}{e_1}{x_2}{e_2};q\steps k\eval [v/x_2]e_2;q}

    \Rule{D-nil}{ }{k\eval\caseS{\nilS}{e_0}{x}{y}{e_1};q\steps k\eval e_0;q}

    \Rule{D-cons}{ }{k\eval\caseS{\consS{v_1}{v_2}}{e_0}{x}{y}{e_1};q\steps k \eval [v_1,v_2/x,y]e_1;q}

    \Rule{D-try}{ }{k\eval \newhandlerS{e}{\ell}{x}{c}{e_\ell}{\Delta}{y}{e_r};q\steps k; \newhandlerS{\_}{\ell}{x}{c}{e_\ell}{\Delta}{y}{e_r} \eval e;q}

    \Rule{D-do}{ }{k\eval\doS{\ell}{v};q\steps k\exnret (\ell, v, \epsilon);q}

    \Rule{D-Capture}{ }{k;x.e\exnret (\ell, v, k');q\steps k\exnret (\ell, v, x.e ~@~ k');q}

    \Rule{D-handle}
    { \ell' \in \Delta }
    {k;\newhandlerS{\_}{\ell}{x}{c}{e_\ell}{\Delta}{y}{e_r} \exnret (\ell', v, k');q\\
    \steps k\eval [v, \dcontS{\newhandlerS{\_}{\ell}{x}{c}{e_\ell}{\Delta}{y}{e_r} ~@~ k'}/x, c]e_{\ell'};q}

    \Rule{D-normal}{ }{k;\newhandlerS{\_}{\ell}{x}{c}{e_\ell}{\Delta}{y}{e_r} \return v;q \steps k \eval [v/y]e_r; q}

    \Rule{D-Dcont}{ }{k \eval \linappS{\dcontS{k'}}{v}; q \steps k ~@~ k' \return v; q}

    \Rule{D-linfun}{ }{k\eval\linappS{\linfunS{x}{e}}{v};q\steps k\eval [v/x]e;q}
  \end{mathpar}

  \subsection{Initial/Final States}
  \framebox{$s \init$ \quad $s \final$}
  \begin{mathpar}
    \Rule{D-init}{ }{\epsilon \eval e;q \init}

    \Rule{D-final}{ }{\epsilon \return v;q \final}

    \Rule{D-final-exn}{ }{\epsilon \exnret (\ell, v, k);q \final}
  \end{mathpar}

  \subsection{Stack Typing}
  \framebox{$\exty{A_1}{\Delta_1} \takes k \takes \exty{A_2}{\Delta_2}$}
  \begin{mathpar}
    \Rule{K-Emp}{ }{\exty B \Delta \takes \epsilon \takes  B \odot \Delta}

    \Rule{K-Bnd}
    {\exty A {\Delta_1} \takes k \takes \exty B {\Delta_2}\\
      x : \tau; q \vdash e : \exty B {\Delta_2}}
    {\exty A {\Delta_1} \takes k; x.e \takes \exty {\aty \tau q} {\Delta_2}}

    \Rule{K-Handler}
    {\exty{A}{\Delta_1} \takes k \takes \exty B {\Delta_3}\\
      \forall \effectT{\ell}{\aty{\tau_1}{q_1}}{C} \in \Delta_2.\
      x : \tau_1, c : \linfunT{C}{B}{\Delta_3}; q_1 \vdash
      e_\ell : \exty{B}{\Delta_3}\\
      y : \tau; q \vdash e_r : B \odot \Delta_3
    }
    {\exty{A}{\Delta_1} \takes k; \newhandlerS{\_}{\ell}{x}{c}{e_\ell}{\Delta_2}{y}{e_r} \takes \exty{\aty \tau q}{\Delta_2}}
  \end{mathpar}

  \subsection{State Typing}
  \framebox{$q \vdash s$}
  \begin{mathpar}
    \Rule{ST-Exp}{\cdot;q \vdash e: \exty B \Delta\\ \_ \takes k \takes \exty B {\Delta}}{q \vdash k \eval e}


    \Rule{ST-Val}{\cdot;q_1 \vdash v:\tau\\ \_ \takes k\takes \exty {\aty \tau {q_2} } \Delta}{q_1 + q_2 \vdash  k\return v}

    \Rule{ST-Eff}
    {
      \cdot;q_1 \vdash v:\tau \\
      \effectT{\ell}{\aty{\tau}{p}}{C} \in \Delta'\\
      \_ \takes k\takes \exty B {\Delta'}\\
      \exty{B}{\Delta'} \takes k' \takes \exty{C}{\Delta''}
    }
    {q_1 + p \vdash k\exnret (\ell, v, k')}
  \end{mathpar}

  \newpage
  \section{Full Soundness Proof}
  \label{sec:appendix-soundness}

  \begin{lemma}
    \label{lem:zero}
    If $\isZero{\tau}$ and $\Phi(v : \tau) = q$, then $q = 0$.
  \end{lemma}

  \begin{lemma}[Sharing]
    If $\tau\curlyveedownarrow(\tau_1,\tau_2)$ and $\emptyctx;q\proves v:\tau$,
    then $\emptyctx;q_1\proves v:\tau_1$, $\emptyctx;q_2\proves
      v:\tau_2$, where $q=q_1+q_2, q_1\geq 0, q_2\geq 0$.
  \end{lemma}
  \begin{proof}
    By simultaneous induction on the types $\tau, \tau_1, \tau_2$.
  \end{proof}

  \begin{lemma}[Potential Relax]
    \label{lem:potrelax}
    If $\Gamma,q\vdash e:\aty{\tau}{p}\odot \Delta$, then $\Gamma,q+r\vdash e:\aty{\tau}{p+r}\odot \Delta$.
  \end{lemma}
  \begin{proof}
    By induction on computation typing rules.
  \end{proof}
  \begin{lemma}[Substitution]
    \label{lem:subst}
    If $\cdot;q_1\vdash v_1:\tau_1$:
    \begin{enumerate}
      \item If $\Gamma,x_1:\tau_1;q_2\vdash v:\tau$, then $\Gamma;q_1+q_2\vdash [v_1/x_1]v:\tau$.
      \item If $\Gamma,x_1:\tau_1;q_2\vdash e:B\odot \Delta$, then $\Gamma;q_1+q_2\vdash
              [v_1/x_1]e:B\odot \Delta$.
    \end{enumerate}
  \end{lemma}

  \begin{proof}
    Induction on $\emptyctx;q_1\proves v_1:\tau_1$:
    \begin{itemize}
      \item \textsc{T-var}: /
      \item \textsc{T-triv, T-fun, T-Pair, T-inl, T-inr, T-nil, T-cons, T-linfun, T-dcont}:\\
            (A): Induction on $\Gamma,x_1:\tau_1;q_2\proves v:\tau$
            \begin{itemize}
              \item \textsc{T-var}: $q_2=0, v=x_1,\Gamma = \emptyctx, \tau_1=\tau$. Then
                    $[v_1/x_1]x_1=v_1$, and
                    $\emptyctx;q_1\proves v_1:\tau_1$.
              \item \textsc{T-triv}: /
              \item \textsc{T-fun}: $q_2=0, v=\funS{f}{x}{e}(x_1\neq x), \tau=\funset$. \\
                    $[v_1/x_1]\funS{f}{x}{e}=\funS{f}{x}{[v_1/x_1]e}$.\\
                    By the assumption, $\forall (\funty{\aty{\tau_x}{q'}}{B\odot \Delta})\in\funset$, we have
                    $\Gamma, x_1:\tau_1,  f:\funset,x:\tau_x;q'\proves e:B\odot \Delta$ and $\abs{\Gamma,x_1:\tau_1}=\Gamma,x_1:\tau_1$ (so $\abs{\tau_1}=\tau_1$).\\
                    Then by lemma \ref{lem:zero},  $q_1=0$.\\
                    Using IH(B), $\forall (\funty{\aty{\tau_x}{q'}}{B\odot \Delta})\in\funset$, we have
                    $\Gamma, f:\funset, x:\tau_x;q'+q_1\proves [v_1/x_1]e:B$, where $q_1=0$, $\abs{\Gamma}=\Gamma$.\\
                    By \textsc{T-fun}, $\Gamma;0\proves\funS{f}{x}{[v_1/x_1]e}:\tau$.

              \item \textsc{T-Pair}:
                    \[\Rule{T-Pair}
                      { \Gamma_1 ; q_1' \vdash v_1' : \tau_1' \\ \Gamma_2; q_2' \vdash v_2' : \tau_2' }
                      { \underbrace{\Gamma_1, \Gamma_2}_{\Gamma, x_1 : \tau_1}; \underbrace{q_1' + q_2'}_{q_2} \vdash \underbrace{\pairS{v_1'}{v_2'}}_v : \underbrace{\prodT{\tau_1'}{\tau_2'}}_\tau}\]
                    \begin{itemize}
                      \item If $x_1 \in \Gamma_1$, using IH(A) on the first premise, $\Gamma_1 \setminus x_1; q_1' + q_1 \proves [v_1/x_1]v_1' : \tau_1'$.\\ By \textsc{T-Pair} on this and the second premise, $\Gamma_1 \setminus x, \Gamma_2; q_1 + q_1' + q_2' \proves \pairS{[v_1/x_1]v_1'}{v_2'} : \prodT{\tau_1'}{\tau_2'}$ as desired.
                      \item If $x_2 \in \Gamma_2$, a symmetric argument can be made, starting with IH(A) on the second premise.
                    \end{itemize}

              \item \textsc{T-inl}:
                    \[
                      \Rule{T-Inl}
                      { \Gamma'; p \vdash v' : \tau_1' }
                      { \underbrace{\Gamma'}_{\Gamma, x_1 : \tau_1}; \underbrace{q_1' + p}_{q_2} \vdash \underbrace{\inlS{v'}}_v : \underbrace{\sumT{\aty{\tau_1'}{q_1'}}{A_2}}_\tau }\]

                    Applying IH(A) on the premise, $\Gamma; p + q_1 \proves [v_1/x_1]v' : \tau_1$.\\
                    By \textsc{T-inl} on this, $\Gamma; p + q_1 + q_1' \proves \inlS{[v_1/x_1]v'} : \sumT{\aty{\tau_1'}{q_1'}}{A_2} $ as deisred.

              \item \textsc{T-inr}: Symmetric to \textsc{T-inl}

              \item \textsc{T-nil}: /

              \item \textsc{T-cons}: $v=\consS{v_1'}{v_2'}$, $\tau=\liT{\aty{\tau_0}{p}}$,
                    $\Gamma,x_1:\tau_1=\Gamma_1,\Gamma_2$.\\
                    By the assumption $q_2=p+q_1'+q_2'$, $\Gamma_1;q_1'\proves
                      v_1':\tau_0$, and $\Gamma_2;q_2'\proves v_2':\liT{\aty{\tau_0}{p}}$.\\
                    Then $x_1:\tau_1\in \Gamma_1$ xor $x_1:\tau_1\in \Gamma_2$.
                    \begin{itemize}
                      \item If $x_1:\tau_1\in \Gamma_1$, using IH(A), $\Gamma_1\setminus(x_1:\tau_1);q_1'+q_1\proves
                              [v_1/x_1]v_1':\tau_0$. \\
                            By \textsc{T-cons}, $\Gamma;q_2+q_1\proves
                              \consS{[v_1/x_1]v_1'}{v_2'}:\liT{\aty{\tau_0}{p}}$.

                      \item Otherwise, $x_1:\tau_1\in \Gamma_2$, using IH(A), $\Gamma_2\setminus(x_1:\tau_1);q_2'+q_1\proves
                              [v_1/x_1]v_2':\liT{\aty{\tau_0}{p}}$.\\
                            By \textsc{T-cons}, $\Gamma;q_2+q_1\proves
                              \consS{v_1'}{[v_1/x_1]v_2'}:\liT{\aty{\tau_0}{p}}$.
                    \end{itemize}

              \item \textsc{T-linfun}: $v=\linfunS{x}{e}(x\neq x_1)$ and
                    $\tau=\linfunT{\aty{\tau'}{p}}{B}{\Delta}$.\\
                    $[v_1/x_1]\linfunS{x}{e}=\linfunS{x}{[v_1/x_1]e}$.\\
                    By the assumption. we have $\Gamma, x_1:\tau_1, x:\tau';q_2+p\proves e:B\odot
                      \Delta$.\\
                    Using IH(B), $\Gamma, x:\tau';q_2+p+q_1\proves [v_1/x_1]e:B\odot \Delta$.\\
                    By \textsc{T-linfun}, $\Gamma;q_2+q_1\proves
                      [v_1/x_1]\linfunS{x}{e}:\linfunT{\aty{\tau'}{p}}{B}{\Delta}$.\\

              \item \textsc{T-dcont}: $v=\dcontS{k}$ has no free variable, so vacuous.\\
            \end{itemize}
            (B): Induction on  $\Gamma,x_1:\tau_1;q_2\proves e:B\odot \Delta$
            \begin{itemize}
              \item \textsc{T-val}: $e=\retS{v'}, B=\aty{\tau}{p}$.\\
                    By the assumption, $\Gamma, x_1:\tau_1;q_2-p\proves v':\tau$.\\
                    Using IH(A), $\Gamma;q_2-p+q_1\proves [v_1/x_1]v':\tau$.\\
                    By \textsc{T-val}, $\Gamma;q_2+q_1\proves
                      [v_1/x_1]\retS{v'}:\aty{\tau}{p}\odot \Delta$.

              \item \textsc{T-let}: $\Gamma, x_1:\tau_1=\Gamma_1,\Gamma_2$ and $e=\letS{e_1}{x}{e_2}$.\\
                    By the assumption, $\Gamma_1;q_2\proves e_1:\aty{\tau}{q_3}\odot \Delta$
                    and $\Gamma_2,x:\tau_1;q_3\proves e_2:B\odot \Delta (x\neq x_1)$\\
                    Then $x_1:\tau_1\in\Gamma_1$ xor $x_1:\tau_1\in\Gamma_2$.
                    \begin{itemize}
                      \item If $x_1:\tau_1\in\Gamma_1$, by IH(B),
                            $\Gamma_1\setminus(x_1:\tau_1);q_2+q_1\proves
                              [v_1/x_1]e_1:\aty{\tau}{q_3}\odot \Delta$.\\
                            By \textsc{T-let}, $\Gamma;q_2+q_1\proves
                              \letS{[v_1/x_1]e_1}{x}{e_2}:B\odot \Delta$.

                      \item Otherwise $x_1:\tau_1\in\Gamma_2$, then by IH(B),
                            $\Gamma_2\setminus(x_1:\tau_1),x:\tau;q_3+q_1\proves
                              [v_1/x_1]e_2:B\odot \Delta$.\\
                            By \textsc{T-let}, and potential relax lemma \ref{lem:potrelax}, $\Gamma;q_2+q_1\proves
                              \letS{e_1}{x}{[v_1/x_1]e_2}:B\odot \Delta$.
                    \end{itemize}

              \item \textsc{T-app}: $\Gamma, x_1:\tau_1=\Gamma_1,\Gamma_2$ and $e=\appS{v_1'}{v_2'}$.\\
                    By the assumption, $q_2=p_1+p_2+p$. $\Gamma_1;p\proves
                      v_1':\funset$, $\funty{\aty{\tau}{p_1}}{B\odot \Delta}\in\funset$ and $\Gamma_2;p_2\proves v_2':\tau$.\\
                    Then $x_1:\tau_1\in\Gamma_1$ xor $x_1:\tau_1\in\Gamma_2$.
                    \begin{itemize}
                      \item If $x_1:\tau_1\in\Gamma_1$, by IH(A), $\Gamma_1\setminus(x_1:\tau_1);q_1+p\proves
                              [v_1/x_1]v_1':\funset$.\\
                            By \textsc{T-app}, $\Gamma;q_1+q_2\proves\appS{[v_1/x_1]v_1'}{v_2'}:B\odot \Delta$.

                      \item Otherwise $x_1:\tau_1\in\Gamma_2$, then by IH(A), $\Gamma_2\setminus(x_1:\tau_1);p_2+q_1\proves
                              [v_1/x_1]v_2':\tau$.\\
                            By \textsc{T-app}, $\Gamma;q_2+q_1\proves\appS{v_1'}{[v_1/x_1]v_2'}:B\odot \Delta$.
                    \end{itemize}

              \item \textsc{T-tick}: /

              \item \textsc{T-casepair}:
                    \[\Rule{T-Casepair}
                      {
                        \Gamma_1; q_1' \vdash v : \prodT{\tau_1'}{\tau_2'}\\
                        \Gamma_2, x_1 : \tau_1', x_2 : \tau_2' ; q_2' \vdash e' : \exty B \Delta
                      }
                      { \underbrace{\Gamma_1, \Gamma_2}_{\Gamma, x_1 : \tau_1}; \underbrace{q_1' +
                          q_2'}_{q_2} \vdash \underbrace{\casepairS{v}{x_1}{x_2}{e'}}_e  : \exty B \Delta}\]

                    \begin{itemize}
                      \item If $x \in \Gamma_1$, by IH(A) on the first premise, $\Gamma_1 \setminus x; q_1' + q_1 \proves [v_1/x_1]v : \prodT{\tau_1'}{\tau_2'}$.\\
                            By \textsc{T-casepair} on this and the second premise, $\Gamma_1 \setminus
                              x, \Gamma_2; q_1' + q_1 + q_2' \proves
                              \casepairS{[v_1/x_1]v}{x_1}{x_2}{e'} : \exty{B}{\Delta}$ as desired.
                      \item If $x \in \Gamma_2$, by IH(B) on the second premise, $\Gamma_2 \setminus
                              x, x_1 : \tau_1', x_2 : \tau_2'; q_2' + q_1 \proves [v_1/x_1]e' :
                              \exty{B}{\Delta}$.\\
                            By \textsc{T-casepair} on this and the first premise, $\Gamma_1 \setminus
                              x, \Gamma_2; q_1' + q_1 + q_2' \proves
                              \casepairS{v}{x_1}{x_2}{[v_1/x_1]e'} : \exty{B}{\Delta}$ as desired.
                    \end{itemize}

              \item \textsc{T-casesum}:
                    Analogous to \textsc{T-caseprod} case.

              \item \textsc{T-caselist}: $\Gamma, x_1:\tau_1=\Gamma_1,\Gamma_2$ and $e=\caseS{v}{e_0}{x}{y}{e_1}$.\\
                    By the assumption, $q_2=q_1'+q_2'$, $\Gamma_1;q_1'\proves v:\liT{\aty{\tau_0}{r}}$,
                    $\Gamma_2;q_2'\proves e_0:B\odot \Delta$,
                    $\Gamma_2,x:\tau_0,y:\liT{\aty{\tau_0}{r}};q_2'+r\proves e_1:B\odot \Delta$.\\
                    Then $x_1:\tau_1\in\Gamma_1$ xor $x_1:\tau_1\in\Gamma_2$.
                    \begin{itemize}
                      \item If $x_1:\tau_1\in\Gamma_1$, by IH(A),
                            $\Gamma_1\setminus(x_1:\tau_1);q_1+q_1'\proves[v_1/x_1]v:
                              \liT{\aty{\tau_0}{r}}$.\\
                            By \textsc{T-caselist}, $\Gamma;q_1+q_2\proves
                              \caseS{[v_1/x_1]v}{e_0}{x}{y}{e_1}:B\odot \Delta$.

                      \item Otherwise $x_1:\tau_1\in\Gamma_2$, then by IH(B),
                            $\Gamma_2\setminus(x_1:\tau_1);q_2'+q_1\proves
                              [v_1/x_1]e_0:B\odot \Delta$, and
                            $\Gamma_2\setminus(x_1:\tau_1),x:\tau_0,y:\liT{\aty{\tau_0}{r}};q_2'+q_1+r\proves
                              [v_1/x_1]e_1:B\odot \Delta$.\\
                            By \textsc{T-caselist}, $\Gamma;q_1+q_2\proves
                              \caseS{v}{[v_1/x_1]e_0}{x}{y}{[v_1/x_1]e_1}:B\odot \Delta$.
                    \end{itemize}

              \item \textsc{T-do}: By the premise, $\Gamma,
                      x_1:\tau_1;p\proves v:\tau_0$ with $\effectT{\ell}{\aty{\tau_0}{q_0}}{B}\in \Delta$, $q_2=p+q_0$.\\
                    Using IH(A), $\Gamma;p+q_1\proves [v_1/x_1]v:\tau_0$.\\
                    By \textsc{T-do}, $\Gamma;q_2+q_1\proves
                      [v_1/x_1]\doS{\ell}{v}:B\odot\Delta$.

              \item \textsc{T-handle}:

                    \[\Rule{T-Handle}
                      { \Gamma_1; q_2 \vdash e : \exty{\aty{\tau}{q}}{\Delta}\\
                        \isZero{\Gamma_2}\\
                        \forall \effectT{\ell}{\aty{\tau_1}{q_l}}{B} \in \Delta.\ (
                        \Gamma_2, x : \tau_1, c : \linfunT{B}{C}{\Delta'}; q_l \vdash
                        e_\ell : \exty{C}{\Delta'})\\\\
                        \Gamma_3, y : \tau; q \vdash e_r : \exty{C}{\Delta'}
                      }
                      { \Gamma_1, \Gamma_2, \Gamma_3; q_2 \vdash \newhandlerS{e}{\ell}{x}{c}{e_\ell}{\Delta}{y}{e_r} :
                        \exty{C}{\Delta'} }\]
                    Let $\Gamma_1, \Gamma_2, \Gamma_3=\Gamma,x_1:\tau_1$. Since $\isZero{\Gamma_2}$, only 2 cases are possible:
                    \begin{itemize}
                      \item If $x_1 \in \Gamma_1$, by IH(B) on the first premise, $\Gamma_1 \setminus
                              x_1; q_2 + q_1 \proves [v_1/x_1]e : \exty{\aty{\tau}{q}}{\Delta}$.\\
                            By \textsc{T-handle}, we have the desired.
                      \item If $x_1 \in \Gamma_2$, by IH(B) on the last premise, $\Gamma_2 \setminus
                              x_1, y : \tau; q + q_1 \proves [v_1/x_1]e_r :
                              \exty{C}{\Delta'}$.\\
                            By \textsc{T-handle} and potential relax lemma \ref{lem:potrelax} on the first premis, we have the desired.
                    \end{itemize}

              \item \textsc{T-ContFun}, \textsc{T-weekfun}, \textsc{T-weakpot}:
                    Follows from applying IH on the premise.

            \end{itemize}
    \end{itemize}
  \end{proof}

  \begin{lemma}[State Relaxing]
    \label{lem:strelax}
    If $q\vdash s$, then $q'\vdash s$ for $q'\geq q$.
  \end{lemma}
  \begin{proof}
    By induction on state and stack typing rules.
  \end{proof}

  \begin{lemma}[Stack Append Typing]
    \label{lem:stack-append}
    $\exty{B_2}{\Delta_2} \takes k_1 ~@~ k_2 \takes \exty{B_1}{\Delta_1}$ if and only if there exists $B, \Delta$ such that $\exty{B_2}{\Delta_2} \takes k_1 \takes \exty{B}{\Delta}$ and $\exty{B}{\Delta} \takes k_2 \takes \exty{B_1}{\Delta_1}$
  \end{lemma}

  \begin{theorem}[Preservation]
    If $q\vdash s$ and $s;q\steps s_0;q_0$ then $q_0\vdash s_0$.
  \end{theorem}

  \begin{proof}
    Induct on $s;q\steps s_0;q_0$.
    \begin{itemize}

      \item \textsc{D-tick}: $k\eval\tickS{r};q\steps k\return \trivS; q_0$ and $q\geq r,
              q_0=q-r$.\\
            $q\proves s$, so $\emptyctx;q\proves \tickS{r}:B\odot \Delta$ and $\_\takes k\takes B\odot \Delta$.\\
            We want to show $q-r\proves k\return \trivS$.\\
            Let $B=\aty{\tau}{x}$. Induct on $\emptyctx;q\proves \tickS{r}:\aty{\tau}{x}\odot \Delta$ to show
            $q-x\geq r$ and
            $\tau=\unitT$. Only 2 cases are possible:
            \begin{itemize}
              \item \textsc{T-tick$^+$}: Then $r\geq 0$, $q=r+x$, $\tau=\unitT$, $x=0$. Good.
              \item \textsc{T-tick$^-$}: Then $r\leq 0$, $q=r+x$, $\tau=\unitT$, $x=0$. Good.


              \item \textsc{T-WeakPot}: By the assumption,
                    $\emptyctx;q'\proves\tickS{r}:\aty{\tau}{x}\odot \Delta$, where $q\geq q'$.\\
                    Using the inner IH $q'-x\geq r$, we have  $q'-x\geq r$.
            \end{itemize}
            By \textsc{T-unit}, $\emptyctx;0\proves \trivS:\unitT$.\\
            We have $\_\takes k\takes \aty{\tau}{x}\odot \Delta$, then by \textsc{ST-val}, $x\proves k\return
              \trivS$.\\
            $x\leq q-r$. By state relaxing lemma \ref{lem:strelax}, $q-r\proves k\return \trivS$.

      \item \textsc{D-fun}: $k\eval\appS{\funS{f}{x}{e}}{v_2};q\steps k\eval
              [\funS{f}{x}{e},v_2/f,x]e;q$.\\
            $q\proves s$, so $\_\takes k\takes B\odot \Delta$, $\emptyctx;q\proves \appS{\funS{f}{x}{e}}{v_2}:B\odot
              \Delta$.\\
            It suffices to show $\emptyctx;q\proves [\funS{f}{x}{e},v_2/f,x]e:B\odot \Delta$.\\
            Induct on $\emptyctx;q\proves \appS{\funS{f}{x}{e}}{v_2}:B\odot \Delta$. Only 2 cases are possible:
            \begin{itemize}
              \item \textsc{T-app}: \\
                    By the assumption, $\emptyctx;0\proves
                      \funS{f}{x}{e}:\funset$, $\emptyctx;q_2\proves v_2:\tau$,
                    $q=q_1+q_2$, $\funty{\aty{\tau}{q_1}}{B\odot \Delta}\in\funset$.\\
                    Invert, get that for $\funty{\aty{\tau}{q_1}}{B\odot \Delta}\in\funset$, we have
                    $f:\funset, x:\tau;q_1\proves e:B\odot \Delta$.\\
                    Using the substitution lemma \ref{lem:subst} twice on $f,x$,
                    $\emptyctx;q_1+q_2\proves [\funS{f}{x}{e},v_2/f,x]e:B\odot \Delta$.


              \item \textsc{T-WeakPot}: \\
                    By the assumption,
                    $\emptyctx;q'\proves\appS{\funS{f}{x}{e}}{v_2}:B\odot \Delta$ and
                    $q\geq q'$.\\
                    Using the inner IH, $\emptyctx;q'\proves
                      [\funS{f}{x}{e},v_2/f,x]e:B\odot \Delta$.\\
                    By \textsc{T-WeakPot}, $\emptyctx;q\proves
                      [\funS{f}{x}{e},v_2/f,x]e:B\odot \Delta$.
            \end{itemize}

      \item \textsc{D-ret}: $k\eval \retS{v};q\steps k\return v;q$.\\
            $q\proves k\eval \retS{v}$, so $\_\takes k\takes B\odot \Delta$ and $\emptyctx;q\proves
              \retS{v}:B\odot \Delta$.\\
            Let $B=\aty{\tau}{r}$. Induct on $\emptyctx;q\proves \retS{v}:\aty{\tau}{r}\odot
              \Delta$
            to show there exists $q_1$, where $\emptyctx;q_1\proves
              v:\tau$ and $q-r\geq q_1$.\\
            Only 2 cases are possible:
            \begin{itemize}
              \item \textsc{T-val}: $B=\aty{\tau}{r}$,  and the assumption gives
                    $\emptyctx;q_1\proves
                      v:\tau$ where $q_1+r=q$.


              \item \textsc{T-WeakPot}:\\
                    By the assumption, $\emptyctx;q\proves \retS{v}:B\odot
                      \Delta$, and $q\geq q'$.\\
                    Using the inner IH, we have $q-r\geq q'-r\geq q_1$, where $\emptyctx;q_1\proves
                      v:\tau$.
            \end{itemize}
            By \textsc{ST-val}, $q_1'+r\proves k\return v$.\\
            $q_1'+r\leq q-r+r=q$, so $q\proves k\return v$ by state relaxing lemma \ref{lem:strelax}.

      \item \textsc{D-let}: $k\eval \letS{e_1}{x}{e_2};q\steps k;x.e_2\eval e_1;q$.\\
            By the assumption, $q\proves s$, so $\emptyctx;q\proves  \letS{e_1}{x}{e_2}:B\odot
              \Delta$,
            $\_\takes k\takes B\odot \Delta$. \\
            Induct on $\emptyctx;q\proves  \letS{e_1}{x}{e_2}:B\odot \Delta$ to show:\\
            for some
            $\tau_1, q_1$, we have $\emptyctx;q\proves e_1:\aty{\tau_1}{q_1}\odot \Delta$ and
            $x:\tau_1;q_1\proves e_2:B\odot \Delta$.\\
            Only 2 cases are possible:
            \begin{itemize}
              \item \textsc{T-let}: Immediate by the assumption.


              \item \textsc{T-WeakPot}:.\\
                    By the assumption, $\emptyctx;q'\proves
                      \letS{e_1}{x}{e_2}:B\odot \Delta$, where $q\geq q'$.\\
                    Using the inner IH, we have $\tau_1, q_1$, such that $\emptyctx;q'\proves
                      e_1:\aty{\tau_1}{q_1}\odot \Delta$,
                    $x:\tau_1;q_1\proves e_2:B\odot \Delta$.\\
                    By \textsc{T-WeakPot}, $\emptyctx;q\proves
                      e_1:\aty{\tau_1}{q_1}\odot \Delta$, $x:\tau_1;q_1\proves e_2:B\odot \Delta$.
            \end{itemize}
            By \textsc{K-bnd}, $k;x.e_2\takes \aty{\tau_1}{q_1}\odot \Delta $.\\
            We have $\emptyctx;q\proves e_1:\aty{\tau_1}{q_1}\odot \Delta$. By \textsc{ST-exp}, $q\proves k;x.e_2\eval e_1$.

      \item \textsc{D-seq}: $k;x.e\return v;q\steps k\eval [v/x]e;q$.\\
            By the assumption, $q\proves  k;x.e\return v$. Invert, get $q=q_1+q_2$,
            $\emptyctx;q_1\proves v:\tau$ and $\_\takes k;x.e\takes \aty{\tau}{q_2}\odot \Delta$.\\
            Invert, $x:\tau;q_2\proves e:B\odot \Delta$ and $\_\takes k\takes B\odot \Delta$.\\
            By the substitution lemma \ref{lem:subst}, $\emptyctx;q_2+q_1\proves [v/x]e:B\odot \Delta$. \\
            By \textsc{ST-exp}, $q_2+q_1\proves k\eval [v/x]e$ ($q=q_1+q_2$).

      \item \textsc{D-nil}: $k\eval \caseS{\nilS}{e_0}{x}{y}{e_1};q\steps k\eval e_0;q$.\\
            By the assumption, $\_\takes k\takes B\odot \Delta$,
            $\emptyctx;q\proves\caseS{\nilS}{e_0}{x}{y}{e_1}:B\odot \Delta$.\\
            It suffices to show $\emptyctx;q\proves e_0:B\odot \Delta$.\\
            Induct on $\emptyctx;q\proves\caseS{\nilS}{e_0}{x}{y}{e_1}:B\odot \Delta$. Only 2 cases are possible:

            \begin{itemize}
              \item \textsc{T-caselist}: By the assumption, $\emptyctx;q_2\proves e_0:B\odot \Delta$,
                    $\emptyctx;q_1\proves \nilS:\liT{A}$ where
                    $q=q_1+q_2$.\\
                    Then by \textsc{T-WeakPot}, $\emptyctx;q\proves e_0:B\odot \Delta$.

              \item \textsc{T-WeakPot}: Analogous to other cases.
            \end{itemize}

      \item \textsc{D-cons}: $k\eval \caseS{\consS{v_1}{v_2}}{e_0}{x}{y}{e_1};q\steps k\eval
              [v_1,v_2/x,y]e_1;q$.\\
            By the assumption, $\_\takes k\takes B\odot \Delta$,
            $\emptyctx;q\proves\caseS{\consS{v_1}{v_2}}{e_0}{x}{y}{e_1}:B\odot \Delta$.\\
            It suffices to show $\emptyctx;q\proves
              [v_1,v_2/x,y]e_1:B\odot \Delta$.\\
            Induct on $\emptyctx;q\proves\caseS{\consS{v_1}{v_2}}{e_0}{x}{y}{e_1}:B\odot \Delta$. Only 2 cases are possible:

            \begin{itemize}
              \item \textsc{T-caselist}: \\
                    By the assumption,
                    $x:\tau,y:\liT{\aty{\tau}{p}};q_2+p\proves  e_0:B\odot \Delta$, where
                    $\emptyctx;q_1\proves \consS{v_1}{v_2}:\liT{\aty{\tau}{p}}$, $q=q_1+q_2$.\\
                    Invert $\emptyctx;q_1\proves \consS{v_1}{v_2}:\liT{\aty{\tau}{p}}$ to get
                    $q_1=q_1'+q_2'+p$, $\emptyctx;q_1'\proves v_1:\tau$, $\emptyctx;q_2'\proves
                      v_2:\liT{\aty{\tau}{p}}$.\\
                    Using substitution lemma \ref{lem:subst}, $\emptyctx;q_2+p+q_1'+q_2'\proves [v_1,v_2/x,y]e_1:B\odot
                      p_0$, which is
                    $\emptyctx;q\proves[v_1,v_2/x,y]e_1:B\odot \Delta$.

              \item \textsc{T-WeakPot}: Analogous to other cases.
            \end{itemize}

      \item \textsc{D-pair}: $k\eval\casepairS{\pairS{v_1}{v_2}}{x_1}{x_2}{e};q\steps k\eval [v_1,v_2/x_1,x_2]e;q$.\\
            By the premise of \textsc{ST-exp}, $k\takes B\odot \Delta$,
            $\emptyctx;q\proves\casepairS{\pairS{v_1}{v_2}}{x_1}{x_2}{e}:B\odot \Delta$.\\
            It suffices to show $\emptyctx;q\proves
              [v_1,v_2/x_1,x_2]e:B\odot \Delta$.\\
            Induct on $\emptyctx;q\proves\casepairS{\pairS{v_1}{v_2}}{x_1}{x_2}{e}:B\odot \Delta$. Only 2 cases are possible:

            \begin{itemize}
              \item \textsc{T-casepair}: \\
                    By the premises,
                    $x_1:\tau_1, x_2 : \tau_2 ;q_2 \proves  e:B\odot \Delta$,
                    $\emptyctx;q_1\proves \pairS{v_1}{v_2}:\prodT{\tau_1}{\tau_2}$, and $q=q_1+q_2$.\\
                    By inversion on $\emptyctx;q_1\proves \pairS{v_1}{v_2}:\prodT{\tau_1}{\tau_2}$, we get
                    $q_1=q_1'+q_2'$, $\emptyctx;q_1'\proves v_1:\tau_1$, and $\emptyctx;q_2'\proves v_2:\tau_2$.\\
                    Using substitution lemma \ref{lem:subst}, $\emptyctx;q_2+q_1'+q_2'\proves [v_1,v_2/x_1,x_2]e:B \odot \Delta$ as desired, since $q = q_2 + q_1' + q_2'$.

              \item \textsc{T-WeakPot}: Analogous to other cases.
            \end{itemize}

      \item \textsc{D-inl}:  $k\eval\casesumS{\inlS{v}}{x_1}{e_1}{x_2}{e_2};q\steps k\eval [v/x_1]e_1;q$.\\
            By the premise of \textsc{ST-exp}, $k\takes B\odot \Delta$,
            $\emptyctx;q\proves \casesumS{\inlS{v}}{x_1}{e_1}{x_2}{e_2} :B\odot \Delta$.\\
            It suffices to show $\emptyctx;q\proves
              [v/x_1]e_1:B\odot \Delta$.\\
            Induct on $\emptyctx;q\proves\casesumS{\inlS{v}}{x_1}{e_1}{x_2}{e_2}:B\odot \Delta$. Only 2 cases are possible:

            \begin{itemize}
              \item \textsc{T-casesum}: \\
                    By the premises,
                    $x_1:\tau_1, ;p_1 + q_2 \proves e_1:B\odot \Delta$,
                    $\emptyctx;q_1\proves \inlS{v} : \sumT{\aty{\tau_1}{p_1}}{A_2}$, and $q=q_1+q_2$.\\
                    By inversion on $\emptyctx;q_1\proves \inlS{v} : \sumT{\aty{\tau_1}{p_1}}{A_2}$, we get
                    $q_1=p_1 + q_1'$ and $\emptyctx;q_1' \proves v:\tau_1$.\\
                    Using substitution lemma \ref{lem:subst}, $\emptyctx;p_1 + q_2 + q_1'\proves [v/x_1]e_1:B \odot \Delta$ as desired, since $q = p_1 + q_2 + q_1'$.

              \item \textsc{T-WeakPot}: Analogous to other cases.
            \end{itemize}

      \item \textsc{D-inr}: Symmetric to above

      \item \textsc{D-do}: $k\eval\doS{\ell}{v};q\steps k\exnret (\ell, v, \epsilon);q$\\
            By the premises of {\sc ST-Exp} on the left, $\cdot; q \vdash \doS{\ell}{v} : \exty{B}{\Delta}$ and $\_ \takes k \takes \exty{B}{\Delta}$.\\
            Induct on $\cdot; q \vdash \doS{\ell}{v} : \exty{B}{\Delta}$ to show that $\cdot; p \vdash v : \tau_1$ with $\effectT{\ell}{\aty{\tau_1}{q_1}}{\aty{\tau_2}{q_2}} \in \Delta$ and $p\leq q-q_1$. Only 2 cases are possible:
            \begin{itemize}
              \item \textsc{T-do}: Follows directly from its premises,  and \(q = p + q_1\).
              \item \textsc{T-WeakPot}: Analogous to other cases.
            \end{itemize}
            Let $k' = \epsilon$. By \textsc{K-Emp} and \textsc{ST-Eff}, $q\vdash k\exnret (l, v, \epsilon)$ as desired.

      \item \textsc{D-Capture}: $k;x.e\exnret (\ell, v, k');q\steps k\exnret (\ell, v, x.e ~@~ k');q$\\
            By the premises of {\sc ST-Eff} on the left, we have
            \begin{itemize}
              \item \(q=q_1+p\)
              \item \(\cdot;q_1 \vdash v:\tau \)
              \item \(\Delta' = \Delta, \effectT{\ell}{\aty{\tau}{p}}{C}\)
              \item \(\_ \takes k; x.e \takes \exty B {\Delta'}\)
              \item \(\exty{B}{\Delta'} \takes k' \takes \exty{C}{\Delta''}\)
            \end{itemize}
            The first 3 premises can be untouched. For the 4th premise, invert with {\sc K-Bnd} to get \(\_ \takes k \takes \exty {B_2} {\Delta'}\) where $x : \tau_2; q_2 \vdash e : \exty{B_2}{\Delta'}$ and $B = \aty{\tau_2}{q_2}$.  For the 5th premise, we first use {\sc K-Emp} and {\sc K-Bnd} to construct that the single-frame stack $x.e$ has typing $\exty{B_2}{\Delta'} \takes x.e \takes \exty{B}{\Delta'}$, then apply lemma \ref{lem:stack-append} between stacks $x.e$ and $k'$ to get that \(\exty{B_2}{\Delta'} \takes x.e ~@~ k' \takes \exty{C}{\Delta''}\).
            \begin{itemize}
              \item \(q=q_1+p\)
              \item \(\cdot;q_1 \vdash v:\tau \)
              \item \(\Delta' = \Delta, \effectT{\ell}{\aty{\tau}{p}}{C}\)
              \item \(\_ \takes k \takes \exty {B_2}{\Delta'}\)
              \item \(\exty{B_2}{\Delta'} \takes x.e ~@~ k' \takes \exty{C}{\Delta''}\)
            \end{itemize}
            Now with these 5 premises, we can apply {\sc ST-Eff} to type the right side.

      \item \textsc{D-Handle}:
            $k;\newhandlerS{\_}{\ell}{x}{c}{e_\ell}{\Delta}{y}{e_r} \exnret (\ell', v, k');q \steps$ $k\eval [v, d/x, c]e_{\ell'};q$\\
            when $\ell \in \Delta$ and where $d = \dcontS{\newhandlerS{\_}{\ell}{x}{c}{e_\ell}{\Delta}{y}{e_r}  ~@~ k'}$.\\
            By the premises of {\sc ST-Eff} on the left, we have
            (1)\ $q=q_1+p$,
            (2)\ $\cdot;q_1 \vdash v:\tau$,
            (3)\ $\effectT{\ell'}{\aty{\tau}{p}}{C} \in \Delta$,
            (4)\ $\_ \takes k; \newhandlerS{\_}{\ell}{x}{c}{e_\ell}{\Delta}{y}{e_r} \takes \exty{B'}{\Delta}$, and
            (5)\ $\exty{B'}{\Delta} \takes k' \takes \exty{C}{\Delta''}$.
            Invert (4) using {\sc K-Handler} to get
            (6) \(\_ \takes k \takes \exty {B} {\Delta_3}\),
            (7) for all $\effectT{\ell}{\aty{\tau_1}{q_1}}{C}$ in $\Delta$,  $x : \tau_1, c : \linfunT{C}{B}{\Delta_3}; q_1 \vdash e_\ell : \exty{B}{\Delta_3}$, and
            (8) $y : \tau'; q' \vdash e_r : \exty{B}{\Delta_3}$ such that $B' = \aty{\tau'}{q'}$.
            Next, we use {\sc K-Emp}, {\sc K-Handler}, (7), (8), and lemma \ref{lem:stack-append} to construct the stack typing \(\exty{B}{\Delta_3} \takes \newhandlerS{\_}{\ell}{x}{c}{e_\ell}{\Delta}{y}{e_r} ~@~ k' \takes \exty{C}{\Delta''}\). Then by {\sc T-Dcont}, we have (9) $\cdot; 0 \vdash d : \linfunT{C}{B}{\Delta_3}$.

            Now, we can apply the substitution lemma for both $x$ and $c$ in $e_\ell$. For $[v/x]$, we have (2). For $[{\sf dcont}(\ldots)/c]$, we have (9). For $e_\ell$, we have (7). Putting it all together, we have
            \[
              \cdot ; p + q_1 + 0 \vdash [v, d/x, c]e' : \exty B {\Delta_3}
            \]
            Finally, using $q = q_1 + p$ from (1), along with \(\_ \takes k \takes \exty {B} {\Delta_3}\) from (6), we can apply {\sc ST-Exp} to type the overall right hand side.

      \item \textsc{D-normal}: $k;\newhandlerS{\_}{\ell}{x}{c}{e_\ell}{\Delta}{y}{e_r} \return v;q\steps
              k\eval [v/y]e_r;q$.\\
            By the assumption, $q=q_1+q_2$, $\emptyctx;q_1\proves v:\tau$, $\_\takes
              k;\newhandlerS{\_}{\ell}{x}{c}{e_\ell}{\Delta}{y}{e_r}\takes
              \aty{\tau}{q_2}\odot \Delta$.\\
            Invert, get $\_\takes k\takes  \aty{\tau}{q_2}\odot \Delta$.\\
            By \textsc{ST-val}, $q_1+q_2=q\proves  k\return v$.

      \item \textsc{D-try}: $k\eval \newhandlerS{e}{\ell}{x}{c}{e_\ell}{\Delta}{y}{e_r};q\steps k;
              \newhandlerS{\_}{\ell}{x}{c}{e_\ell}{\Delta}{y}{e_r} \eval e;q$.\\
            By the assumption, $\_\takes k\takes B\odot \Delta$ and
            $\emptyctx;q\proves\newhandlerS{e}{\ell}{x}{c}{e_\ell}{\Delta}{y}{e_r}:B\odot \Delta$.\\
            Induct on
            $\emptyctx;q\proves\newhandlerS{e}{\ell}{x}{c}{e_\ell}{\Delta}{y}{e_r}:B\odot
              \Delta$ to show there exists $\Delta', q'$ such that $q \ge q'$, and
            $\emptyctx;q'\proves e_1:B\odot \Delta'$, and $\forall
              \effectT{\ell}{\aty{\tau_1}{q_1}}{\aty{\tau_2}{q_2}} \in \Delta', x :
              \tau_1, c : \linfunT{\aty{\tau_2}{q_2}}{A}{\Delta}; q_1 \vdash e_\ell
              : \exty{A}{\Delta}$, and $y:\tau\proves e_r:B\odot\Delta$\\
            Only 2 cases are possible:
            \begin{itemize}
              \item \textsc{T-Handle}: Immediate by the assumption when context is
                    empty, and choosing $q' = q$.


              \item \textsc{T-WeakPot}: By the assumption, for some $q'''$ such that
                    $q \ge q'''$, we have
                    $\emptyctx;q'''\proves\newhandlerS{e}{\ell}{x}{c}{e_\ell}{\Delta}{y}{e_r}:B\odot
                      \Delta$. Then apply the inner IH to obtain that there exists
                    $\Delta', q''$ such that $q''' \ge q''$, and
                    $\emptyctx;q'\proves e_1:B\odot \Delta'$, and $\forall
                      \effectT{\ell}{\aty{\tau_1}{q_1}}{\aty{\tau_2}{q_2}} \in
                      \Delta', x : \tau_1, c :
                      \linfunT{\aty{\tau_2}{q_2}}{A}{\Delta}; q_1 \vdash e_\ell :
                      \exty{A}{\Delta}$, and $y:\tau\proves e_r:B\odot\Delta$.
                    Now choose $q' = q''$, then by transitivity $q \ge q'$. The
                    remaining conditions immediately hold.
            \end{itemize}
            Then by \textsc{K-handler} and \textsc{ST-exp} and \textsc{T-WeakPot},
            $q\proves
              k;\newhandlerS{\_}{\ell}{x}{c}{e_\ell}{\Delta}{y}{e_r}\eval e$.

      \item \textsc{D-Dcont}: $k \eval \appS{\dcontS{k'}}{v}; q \steps k ~@~ k' \return v; q$\\
            By the premises of {\sc ST-Exp} on the left, $\cdot;q \vdash \appS{\dcontS{k'}}{v} : \exty B \Delta$ where $\_ \takes k \takes \exty B {\Delta}$. We want to show that $\cdot; q_1 \vdash v : \tau$ and $\_ \takes k ~@~ k' \takes \exty{\aty{\tau}{q_2}}{\Delta'}$ where $q = q_1 + q_2$.\\ Induct on $\cdot;q \vdash \appS{\dcontS{k'}}{v} : \exty B \Delta$. Only 2 cases are possible:
            \begin{itemize}
              \item \textsc{T-linapp}: By the premises, we have $\cdot; 0 \vdash \dcontS{k'} : \linfunT{\aty{\tau}{q_2}}{B}{\Delta}$ and $\cdot; q_1 \vdash v : \tau$ where $q = q_1 + q_2$. Invert on the first premise with {\sc T-Dcont} to get that $\exty{B}{\Delta} \takes k' \takes \exty{\aty{\tau}{q_2}}{\Delta'}$. Then by lemma \ref{lem:stack-append} between the typings of $k$ and $k'$, we have that $\_ \takes k ~@~ k' \takes \exty{\aty{\tau}{q_2}}{\Delta'}$.

              \item \textsc{T-WeakPot}: Analogous to other cases.
            \end{itemize}
            By {\sc ST-Val}, we get the right side as desired.
    \end{itemize}
  \end{proof}

  \begin{lemma}[Canonical Forms]
    If $\emptyctx;q\proves v:\tau$, and
    \begin{itemize}
      \item $\tau=\liT{A}$, then either $v=\nilS$ or $v=\consS{v_1}{v_2}$.
      \item $\tau = \prodT{\tau_1}{\tau_2}$, then $v = \pairS{v_1}{v_2}$
      \item $\tau = \sumT{A_1}{A_2}$, then either $v = \inlS{v_1}$ or $v = \inrS{v_2}$
      \item $\tau=\funset$, then $v=\funS{f}{x}{e'}$, or $\linfunS{x}{e}$, or $\dcontS{k}$.
    \end{itemize}
  \end{lemma}
  \begin{corollary}
    \label{lemma:canon-corollary}
    If $\emptyctx;q\proves e:B\odot \Delta$, and
    \begin{itemize}
      \item $e=\caseS{v}{e_0}{x}{y}{e_1}$, then either $v=\nilS$ or $v=\consS{v_1}{v_2}$.
      \item $e=\casepairS{v}{x_1}{x_2}{e'}$, then $v = \pairS{v_1}{v_2}$
      \item $e=\casesumS{v}{x_1}{e_1}{x_2}{e_2}$, then either $v = \inlS{v_1}$ or $v = \inrS{v_2}$
      \item $e=\appS{v_1}{v_2}$, then $v_1=\funS{f}{x}{e'}$, or $\linfunS{x}{e'}$, or $\dcontS{k}$.
    \end{itemize}
  \end{corollary}

  \begin{theorem}[Progress]
    If $q\proves s$, then either $s$ final or $s;q\steps s';q'$.
  \end{theorem}

  \begin{proof}
    Induct on $q\proves s$.
    \begin{itemize}
      \item \textsc{ST-val}: $q\proves k\return v$.\\
            Induct on structure of k.
            \begin{itemize}
              \item $k=\emptystk$: $k\return v$ is final.
              \item $k=k';x.e$: By \textsc{D-seq}, $k';x.e\return v;q\steps
                      k'\eval[v/x]e;q$.
              \item $k=k';\newhandlerS{\_}{\ell}{x}{c}{e_\ell}{\Delta}{y}{e_r};$: By \textsc{D-normal},
                    $k';\newhandlerS{\_}{\ell}{x}{c}{e_\ell}{\Delta}{y}{e_r}\return
                      v;q\steps k'\eval [v/y]e_r;q$.
            \end{itemize}

      \item {\sc ST-Eff}: $q \vdash k\exnret (\ell', v, k')$\\
            Induct on structure of k.
            \begin{itemize}
              \item $k=\emptystk$: $k \exnret (\ell', v, k')$ is final

              \item $k=k_1;x.e$: By \textsc{D-Capture}, $k_1;x.e\exnret (\ell', v, k');q \steps k_1 \exnret (\ell', v, x.e ~@~ k');q$

              \item $k=k_1;\newhandlerS{\_}{\ell}{x}{c}{e_\ell}{\Delta}{y}{e_r}$: Our type system
                    guarantees that $\ell' \in \Delta$. Then by \textsc{D-handle}, $k_1;\newhandlerS{\_}{\ell}{x}{c}{e_\ell}{\Delta}{y}{e_r}\exnret (\ell', v, k');q\\
                      \steps k_1 \eval [v,
                        \dcontS{\newhandlerS{\_}{\ell}{x}{c}{e_\ell}{\Delta}{y}{e_r} ~@~ k'}/x, c]e_{\ell'};q$
            \end{itemize}

      \item \textsc{ST-exp}: $q\proves k\eval e$.\\
            By the premises, $\emptyctx;q\proves e:B\odot \Delta$ and $\_ \takes k\takes B\odot \Delta$. Induct on the structure of $e$.
            \begin{itemize}
              \item $e=\caseS{v}{e_0}{x}{y}{e_1}$\\
                    By type inversion and the canonical forms lemma, either $v=\nilS$ or $v=\consS{v_1}{v_2}$.
                    \begin{itemize}
                      \item $v = \nilS$: by \textsc{D-nil}, $k\eval \caseS{\nilS}{e_0}{x}{y}{e_1};q\steps k\eval e_0;q$
                      \item $v=\consS{v_1}{v_2}$: by \textsc{D-cons}, $k\eval \caseS{\consS{v_1}{v_2}}{e_0}{x}{y}{e_1};q\steps k\eval
                              [v_1,v_2/x,y]e_1;q$
                    \end{itemize}

              \item $e$ is some other non-$\tickS{r}$ expression. These cases follow similarly to $e=\caseS{v}{e_0}{x}{y}{e_1}$. For each, first apply the canonical forms lemma/corollary, followed by one of the rules: \textsc{D-ret},
                    \textsc{D-fun},\textsc{D-let}, \textsc{D-seq}, \textsc{D-pair}, \textsc{D-inl}, \textsc{D-inr},
                    \textsc{D-do}, \textsc{D-try}, \textsc{D-linfun}, \textsc{D-Dcont}. This will end up showing that $k\eval e;q\steps s';q$ for some $s'$

              \item $e=\tickS{r}$.\\
                    Induct on $\emptyctx;q\proves \tickS{r}:B\odot \Delta$ to show $q\geq r$. Only 2 cases are possible:
                    \begin{itemize}
                      \item \textsc{T-tick}: $q=r$.


                      \item \textsc{T-WeakPot}: $q\geq r$ by the IH.
                    \end{itemize}
                    $q\geq r$ then by \textsc{D-tick}, $k\eval\tickS{r};q\steps k\return
                      \trivS;q-r$.
            \end{itemize}

    \end{itemize}
  \end{proof}
\end{appendices}